\documentclass[runningheads]{llncs}

\usepackage[T1]{fontenc}
\usepackage{amsmath}
\usepackage{amssymb}
\usepackage{graphicx}
\usepackage{hyperref}

\usepackage{amssymb}%
\usepackage{pifont}%
\usepackage{mathtools}
\usepackage{booktabs}
\usepackage{xcolor}
\usepackage{multirow}
\usepackage[linesnumbered,lined,ruled,noend]{algorithm2e}
\usepackage[capitalize]{cleveref}
\usepackage{tabularx}
\usepackage{todonotes}
\usepackage{cite}
\usepackage{gnuplot-lua-tikz}
\usetikzlibrary{calc}
\usepackage{fancyvrb}

\usepackage{versions}
\includeversion{ExtendedVersion}

\newcommand{\crefAppendix}[1]{the full version\xspace}%
\begin{ExtendedVersion}
\renewcommand{\crefAppendix}[1]{\cref{appendix:#1}}
\end{ExtendedVersion}

\DontPrintSemicolon
\SetKwInOut{Input}{Input}
\SetKwInOut{Output}{Output}
\SetKw{Outputs}{output}

\newcommand{\SrcURL}{\url{https://doi.org/10.5281/zenodo.6558657}\xspace}
\newcommand{\AP}{\mathrm{AP}}
\newcommand{\SK}{\mathrm{SK}}
\newcommand{\BK}{\mathrm{BK}}
\newcommand{\PK}{\mathrm{PK}}

\newcommand{\Enc}{\textrm{Enc}}
\newcommand{\Dec}{\textrm{Dec}}

\newcommand{\defiff}{\stackrel{\mathrm{def}}{\iff}}
\newcommand{\cequiv}{\stackrel{\mathsf{c}}{\equiv}}
\newcommand{\Rgets}{\stackrel{\mathsf{R}}{\gets}}
\newcommand{\X}{\mathsf{X}}
\newcommand{\G}{\mathsf{G}}
\newcommand{\F}{\mathsf{F}}
\newcommand{\B}{\mathbb{B}}
\newcommand{\N}{\mathbb{N}}
\newcommand{\CMux}{\textsc{CMux}\xspace}
\newcommand{\LookUp}{\textsc{LookUp}\xspace}
\newcommand{\SampleExtract}{\textsc{SampleExtract}\xspace}
\newcommand{\Bootstrapping}{\textsc{Bootstrapping}\xspace}
\newcommand{\CircuitBootstrapping}{\textsc{CircuitBootstrapping}\xspace}
\newcommand{\Trivial}{\textsc{Trivial}\xspace}

\newcommand{\TLWE}{\textsf{TLWE}\xspace}
\newcommand{\TRLWE}{\textsf{TRLWE}\xspace}
\newcommand{\TRGSW}{\textsf{TRGSW}\xspace}
\newcommand{\FHOffline}{\textsc{Offline}\xspace}
\newcommand{\ReverseStream}{\textsc{Reverse}\xspace}
\newcommand{\BlockStream}{\textsc{Block}\xspace}
\newcommand{\BootstrappingInterval}{I_{\mathrm{boot}}}
\newcommand{\OutputInterval}{I_{\mathrm{out}}}
\newcommand{\Rev}[1]{#1^{\mathrm{R}}}

\newcommand{\IdentityKeySwitching}{\textsc{IdentityKeySwitching}\xspace}
\newcommand{\PrivateKeySwitching}{\textsc{PrivateKeySwitching}\xspace}
\newcommand{\TLWElvlz}{\textsf{TLWElvl0}\xspace}
\newcommand{\TLWElvlo}{\textsf{TLWElvl1}\xspace}
\newcommand{\TLWElvlt}{\textsf{TLWElvl2}\xspace}
\newcommand{\TRLWElvlo}{\textsf{TRLWElvl1}\xspace}
\newcommand{\TRLWElvlt}{\textsf{TRLWElvl2}\xspace}
\newcommand{\TRGSWlvlo}{\textsf{TRGSWlvl1}\xspace}
\newcommand{\TRGSWlvlt}{\textsf{TRGSWlvl2}\xspace}

\newif\ifdraft\draftfalse

\def\mkDraftFn#1#2{%
  \expandafter\def\csname #1\endcsname##1{%
  \ifdraft\todo[%
    color=#2%
  ]{#1: ##1}\fi}%
}
\mkDraftFn{KS}{red!40}
\mkDraftFn{MW}{blue!40}
\mkDraftFn{SB}{orange!40}
\mkDraftFn{RB}{green!40}
\mkDraftFn{KM}{gray!40}

\usepackage{ifthen}
\newcommand{\colorR}[1]{\textcolor{red}{#1}}
\newcommand{\pagelimitmarker}[1]{~\\ {\colorR{\ifthenelse{\thepage>#1}{\Huge Exceeding the page limit}{\huge Within the page limit}}}~\\ {\huge{\colorR{~~Page Limit\,\,\,\,\, = #1}}}~\\ {\huge{\colorR{~~Current Page = $\thepage$}}}}

\usepackage[]{lineno} %

\usepackage{paralist} %

\newenvironment{oneenumeration}
	{\begin{inparaenum}[1)]}
	{\end{inparaenum}}

\makeatletter
\def\orcidID#1{\smash{\href{https://orcid.org/#1}{\protect\raisebox{-1.25pt}{\protect\includegraphics{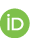}}}}}
\makeatother

\usepackage{subcaption}
\begin{document}
\VerbatimFootnotes
\title{%
Oblivious Online Monitoring for Safety LTL Specification via Fully Homomorphic Encryption%
}
\titlerunning{%
Oblivious Online Monitoring for Safety LTL Specification via FHE}
\author{%
    Ryotaro Banno\inst{1}\orcidID{0000-0001-8610-9186}\and
    Kotaro Matsuoka\inst{1}\orcidID{0000-0002-6642-1276}\and
    Naoki Matsumoto\inst{1}\orcidID{0000-0002-4497-3459}\and
    Song Bian\inst{2}\orcidID{0000-0003-0467-6203}\and
    Masaki Waga\inst{1}\orcidID{0000-0001-9360-7490}\and
    Kohei Suenaga\inst{1}\orcidID{0000-0002-7466-8789}}
\authorrunning{R. Banno et al.}
\institute{%
    Kyoto University, Kyoto, Japan%
    \and
    Beihang University, Beijing, China%
}
\maketitle              %
\ifdraft\linenumbers{}
\fi
\begin{abstract}

In many Internet of Things (IoT) applications, data sensed by an IoT device are continuously sent to the server and monitored against a specification. Since the data often contain sensitive information, and the monitored specification is usually proprietary, both must be kept private from the other end.
We propose a protocol to conduct \emph{oblivious online monitoring}---online monitoring conducted without revealing the private information of each party to the other---against a safety LTL specification.
In our protocol, we first convert a safety LTL formula into a DFA and conduct online monitoring with the DFA. Based on \emph{fully homomorphic encryption (FHE)}, we propose two online algorithms (\ReverseStream{} and \BlockStream{}) to run a DFA obliviously. %
We prove the correctness and security of our entire protocol.
We also show the scalability of our algorithms theoretically and empirically.
Our case study shows that our algorithms are fast enough to monitor blood glucose levels online, 
demonstrating our protocol's practical relevance.

\end{abstract}
\section{Introduction}
\begin{figure}[t]
    \centering
    \includegraphics[width=0.85\textwidth]{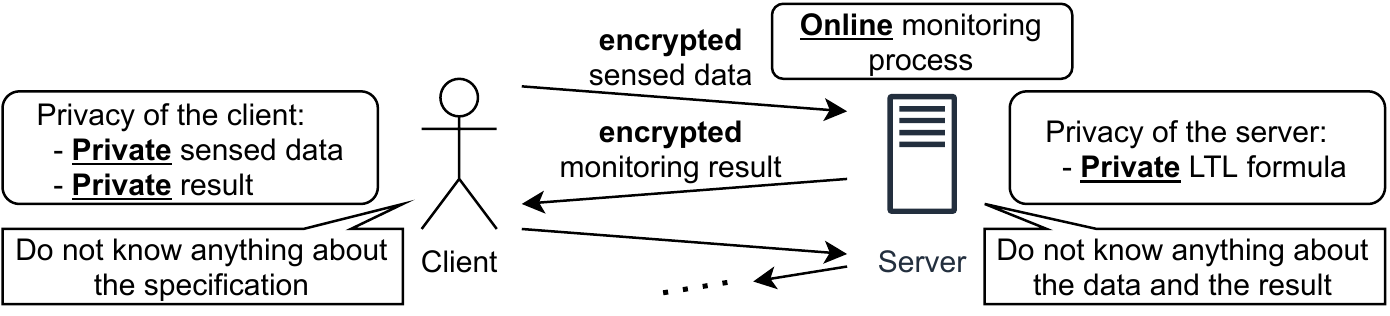}
    \caption{The proposed oblivious online LTL monitoring protocol.}
    \label{fig:motivation}
\end{figure}
Internet of Things (IoT)~\cite{atzori2010internet}
devices enable various service providers to monitor personal data of their users
and to provide useful feedback to the users. %
For example, a smart home system can save lives 
by raising an alarm when a gas stove is left on to prevent a fire. %
Such a system is realized by the
continuous monitoring of the data from the IoT devices in the house~\cite{bing2011design,DBLP:conf/rv/El-HokayemF18a}.
Another application of IoT devices is medical IoT (MIoT)~\cite{dimitrov2016medical}. 
In MIoT applications, biological information,
such as electrocardiograms or blood glucose levels, is monitored, and the user
is notified when an abnormality is detected (such as arrhythmia or hyperglycemia).

In many IoT applications\begin{ExtendedVersion}~\cite{DBLP:journals/fgcs/BottaDPP16}\end{ExtendedVersion}, monitoring must be conducted \emph{online}, i.e., a stream of sensed data is continuously monitored, and the violation of the monitoring specification must be reported even before the entire data are obtained.
In the smart home and MIoT applications, online 
monitoring is usually required, as continuous sensing
is crucial for the immediate notifications to emergency responders,
such as police officers or doctors, for the ongoing
abnormal situations.

As specifications generally contain proprietary information or sensitive
parameters learned from private data (e.g., with specification mining~\cite{DBLP:conf/kbse/LemieuxPB15}),
\emph{the specifications must be kept secret}.
One of the approaches for this privacy is to adopt the client-server model to the monitoring system. 
In such a model, the sensing device sends the collected data to a server, where the server performs the necessary analyses and returns the results to the device.
Since the client does not have access to the specification, the server's privacy is preserved.

However, the client-server model does \emph{not} inherently protect the client's privacy from the servers, as the data collected from and results sent back to the users are revealed to the servers in this model; that is to say, a user has to \emph{trust} the server.
This trust is problematic if, for example, the server itself intentionally or unintentionally leaks sensitive data of device users to an unauthorized party.
Thus, we argue that a monitoring procedure should achieve the following goals:
\begin{description}
    \item[Online monitoring.] The %
    monitored data need not be known beforehand. %
    \item[Client's Privacy.]  The server shall not know the monitored data and results. %
    \item[Server's Privacy.]  The client shall not know what property is monitored. %
\end{description}
We call a monitoring scheme with these properties \emph{oblivious online monitoring}.
By an oblivious online monitoring procedure,
\begin{oneenumeration}
 \item a user can get a monitoring result hiding her sensitive data and the result itself from a server, and
 \item a server can conduct online monitoring hiding the specification from the user.
\end{oneenumeration}

\begin{figure}[tbp]
 \tikzset{textblock/.style={rectangle, draw, text centered,node distance=0cm, text width=1cm}}
 \begin{subfigure}[b]{0.40\textwidth}
 \centering
 \begin{tikzpicture}[scale=0.65,every node/.style={transform shape}]
  \node at (0,0) (dots) {$\cdots$};
  \node[textblock, left=of dots] (d2) {$d_2$};
  \node[textblock, left=of d2] (d1) {$d_1$};
  \node[textblock, right=of dots] (dn1) {$d_{n-1}$};
  \node[textblock, right=of dn1] (dn) {$d_{n}$};

  \def\diffLine{0.55}
  \draw [thick,|-|] ($(d1.west) - (0, \diffLine)$)--($(dn.east) - (0, \diffLine)$);
  \draw[fill=white]
  ($(d2) + (0.2, -\diffLine)$) --
  ($(d2.east)  - (0, \diffLine + 0.2)$) --
  ($(dn1.west) - (0, \diffLine + 0.2)$) --
  ($(dn1.west) - (0, \diffLine - 0.2)$) --
  ($(d2.east)  - (0, \diffLine - 0.2)$) --
  ($(d2) + (0.2, -\diffLine)$)
  ;
  \node[] at ($(dots) - (0, \diffLine)$) {$M$};
 \end{tikzpicture}
  \caption{Algorithm \FHOffline{}.}
  \label{fig:outline_algorithm:offline}
 \vspace{1em}
 \begin{tikzpicture}[scale=0.65,every node/.style={transform shape}]
  \node at (0,0) (revdots) {$\cdots$};
  \node[textblock, right=of revdots] (revd2) {$d_{\color{red}\mathbf{2}}$};
  \node[textblock, right=of revd2] (revd1) {$d_{\color{red}\mathbf{1}}$};
  \node[textblock, left=of revdots] (revdn1) {$d_{\color{red}\mathbf{n-1}}$};
  \node[textblock, left=of revdn1] (revdn) {$d_{\color{red}\mathbf{n}}$};

  \def\diffLine{0.55}
  \draw [thick,|-|] ($(revd1.east) - (0, \diffLine)$)--($(revdn.west) - (0, \diffLine)$);
  \draw[fill=white]
  ($(revdn1) + (0.2, -\diffLine)$) --
  ($(revdn1.east) - (0, \diffLine - 0.2)$) --
  ($(revd2.west)  - (0, \diffLine - 0.2)$) --
  ($(revd2.west)  - (0, \diffLine + 0.2)$) --
  ($(revdn1.east) - (0, \diffLine + 0.2)$) --
  ($(revdn1) + (0.2, -\diffLine)$)
  ;
  \node[] at ($(revdots) - (0, \diffLine)$) {$\Rev{M}$};

  \node at ($(revdn.west) - (0, 2.5 * \diffLine)$) {\Large $=$};

  \node at ($(revdots.east) - (0, 2 * \diffLine)$) (dots) {$\cdots$};
  \node[textblock, left=of dots] (d2) {$d_2$};
  \node[textblock, left=of d2] (d1) {$d_1$};
  \node[textblock, right=of dots] (dn1) {$d_{n-1}$};
  \node[textblock, right=of dn1] (dn) {$d_{n}$};

  \def\diffLine{0.55}
  \draw [thick,|-|] ($(d1.west) - (0, \diffLine)$)--($(dn.east) - (0, \diffLine)$);
  \draw[fill=white]
  ($(dn1) + (-0.2, -\diffLine)$) --
  ($(dn1.west)  - (0, \diffLine + 0.2)$) --
  ($(d2.east) - (0, \diffLine + 0.2)$) --
  ($(d2.east) - (0, \diffLine - 0.2)$) --
  ($(dn1.west)  - (0, \diffLine - 0.2)$) --
  ($(dn1) + (-0.2, -\diffLine)$)
  ;
  \node[] at ($(dots) - (0, \diffLine)$) {$\Rev{M}$};
 \end{tikzpicture}
 \caption{Algorithm \ReverseStream{}, where $\Rev{M}$ is the reversed DFA of $M$.}
 \label{fig:outline_algorithm:reversed}
 \end{subfigure}
 \hfill
 \begin{subfigure}[b]{0.57\textwidth}
 \centering
 \begin{tikzpicture}[scale=0.65,every node/.style={transform shape}]
  \node at (0,0) (dots) {$\cdots$};
  \node[textblock, left=of dots] (d4) {$d_4$};
  \node[textblock, left=of d4] (d3) {$d_3$};
  \node[textblock, left=of d3] (d2) {$d_2$};
  \node[textblock, left=of d2] (d1) {$d_1$};
  \node[textblock, right=of dots] (dn1) {$d_{n-1}$};
  \node[textblock, right=of dn1] (dn) {$d_{n}$};

  \def\diffLine{0.6}
  \draw [thick,|-|] ($(d1.west) - (0, \diffLine)$)--($(d2.east) - (0, \diffLine)$);
  \draw[fill=white]
  ($(d1) + (0, -\diffLine)$) --
  ($(d1.east)  - (0.2, \diffLine + 0.2)$) --
  ($(d2.west) - (-0.3, \diffLine + 0.2)$) --
  ($(d2.west) - (-0.3, \diffLine - 0.2)$) --
  ($(d1.east)  - (0.2, \diffLine - 0.2)$) --
  ($(d1) + (0, -\diffLine)$)
  ;
  \node[] at ($(d1.east) - (0, \diffLine)$) {$M$};
  \draw [dashed,thick] ($(d1.west) - (0, \diffLine + 0.2)$) edge[bend right] node[below] {$B$} ($(d2.east) - (0, \diffLine + 0.2)$);

  \def\diffLineSecond{2 * \diffLine}
  \draw [thick,|-|] ($(d3.west) - (0, \diffLineSecond)$)--($(d4.east) - (0, \diffLineSecond)$);
  \draw[fill=white]
  ($(d3) + (0, -\diffLineSecond)$) --
  ($(d3.east)  - (0.2, \diffLineSecond + 0.2)$) --
  ($(d4.west) - (-0.3, \diffLineSecond + 0.2)$) --
  ($(d4.west) - (-0.3, \diffLineSecond - 0.2)$) --
  ($(d3.east)  - (0.2, \diffLineSecond - 0.2)$) --
  ($(d3) + (0, -\diffLineSecond)$)
  ;
  \node[] at ($(d3.east) - (0, \diffLineSecond)$) {$M$};
  \draw [dashed,thick] ($(d3.west) - (0, \diffLineSecond + 0.2)$) edge[bend right] node[below] {$B$} ($(d4.east) - (0, \diffLineSecond + 0.2)$);

  \node at ($(0,0) - (0, 2.5 * \diffLine)$) (dots) {$\ddots$};

  \def\diffLineThird{3.5 * \diffLine}
  \draw [thick,|-|] ($(dn1.west) - (0, \diffLineThird)$)--($(dn.east) - (0, \diffLineThird)$);
  \draw[fill=white]
  ($(dn1) + (0, -\diffLineThird)$) --
  ($(dn1.east)  - (0.2, \diffLineThird + 0.2)$) --
  ($(dn.west) - (-0.3, \diffLineThird + 0.2)$) --
  ($(dn.west) - (-0.3, \diffLineThird - 0.2)$) --
  ($(dn1.east)  - (0.2, \diffLineThird - 0.2)$) --
  ($(dn1) + (0, -\diffLineThird)$)
  ;
  \node[] at ($(dn1.east) - (0, \diffLineThird)$) {$M$};
  \draw [dashed,thick] ($(dn1.west) - (0, 3.5 * \diffLine + 0.2)$) edge[bend right] node[below] {$B$} ($(dn.east) - (0, 3.5* \diffLine + 0.2)$);
 \end{tikzpicture}
  \caption{Algorithm \BlockStream{} with block size $B = 2$. Each block of length $B$ is consumed with a variant of \FHOffline{}. The intermediate result at each block is used in the consumption of the next block.}
  \label{fig:outline_algorithm:bbs}
 \end{subfigure}
 \caption{How our algorithms consume the data $d_1, d_2,\dots,d_n$ with the DFA $M$.}
 \label{fig:outline_algorithm}
\end{figure}
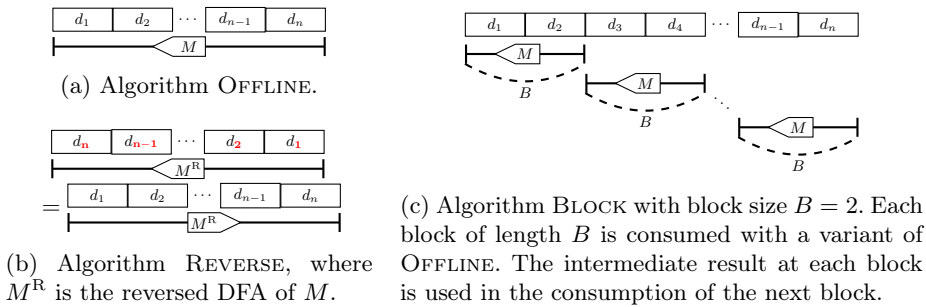
\paragraph{Contribution.}

In this paper, we propose a novel protocol (\cref{fig:motivation}) for oblivious online monitoring against a specification in \emph{linear temporal logic (LTL)}~\cite{DBLP:conf/focs/Pnueli77}. %
More precisely, we use a \emph{safety LTL formula}~\cite{DBLP:journals/fmsd/KupfermanV01} as a specification, which can be translated to a deterministic finite automaton (DFA)~\cite{DBLP:journals/fmsd/TabakovRV12}.
In our protocol, we first convert a safety LTL formula into a DFA and conduct online monitoring with the DFA.\@
For online and oblivious execution of a DFA, we propose two algorithms %
 based on \emph{fully homomorphic encryption} (FHE).
FHE allows us to evaluate an arbitrary function over ciphertexts, and there is an FHE-based algorithm to evaluate a DFA obliviously~\cite{DBLP:journals/joc/ChillottiGGI20}.
However, this algorithm
is limited to \emph{leveled} homomorphic, i.e., the FHE parameters are dependent on the number of the monitored ciphertexts and thus not applicable to online monitoring.

In this work,
we first present a \emph{fully} homomorphic \emph{offline} DFA evaluation algorithm (\FHOffline) by extending the leveled homomorphic algorithm in~\cite{DBLP:journals/joc/ChillottiGGI20}.
Although we can remove the parameter dependence using this method, \FHOffline consumes the ciphertexts  from back to front (\cref{fig:outline_algorithm:offline}). As a result, \FHOffline{} is still limited to offline usage only. %
To truly enable online monitoring, we propose two new algorithms based on \FHOffline: \ReverseStream{} and \BlockStream{}.
In \ReverseStream{}, we \emph{reverse} the DFA and apply \FHOffline{} to the reversed DFA (\cref{fig:outline_algorithm:reversed}).\@
In \BlockStream{}, we split the monitored ciphertexts into fixed-length \emph{blocks} and process each block sequentially with \FHOffline{} (\cref{fig:outline_algorithm:bbs}).
We prove that both of the algorithms have \emph{linear} time complexity and \emph{constant} space complexity to the length of the monitored ciphertexts, which guarantees the scalability of our entire protocol.

On top of our online algorithms, we propose a protocol for oblivious online LTL monitoring.
We assume that the client is \emph{malicious}, i.e., the client can deviate arbitrarily
from the protocol, while the server is \emph{honest-but-curious}, i.e.,
the server honestly follows the protocol but tries to learn the client's private data
by exploiting the obtained information.
We show that the privacy of both parties can be protected under the standard IND-CPA security of FHE schemes with the addition
of \emph{shielded randomness leakage} (SRL) 
security~\cite{DBLP:journals/iacr/BrakerskiDGM20a,DBLP:journals/iacr/GayP20}.

We implemented our algorithms for DFA evaluation in C++20
and evaluated their performance.
Our experiment results confirm the scalability of our algorithms.
Moreover, through a case study on blood glucose levels monitoring, we also show that our algorithms run fast enough for online monitoring, i.e., our algorithms are faster than the sampling interval of the current commercial devices that samples glucose levels.

Our contributions are summarized as follows:
\begin{itemize}
    \item We propose two \emph{online} algorithms to run a DFA obliviously.
    \item We propose the first protocol for oblivious online LTL monitoring. 
    \item We proved the correctness and security of our protocol.
    \item Our experiments show the scalability and practicality of our algorithms. %
\end{itemize}

\noindent{\textit{Related work.}}
There are various works on DFA execution without revealing the monitored data (See \cref{tbl:summary-related-work} for a summary). However, to our knowledge, there is no existing work achieving all of our three goals (i.e., \emph{online monitoring}, \emph{privacy of the client}, and \emph{privacy of the server}) simultaneously. Therefore, none of them is applicable to oblivious online LTL monitoring.

Homomorphic encryption, which we also utilize, has been used to run a DFA obliviously~\cite{DBLP:journals/joc/ChillottiGGI20,DBLP:conf/tcc/IshaiP07}. Among different homomorphic encryption schemes, our algorithm is based on the algorithm in~\cite{DBLP:journals/joc/ChillottiGGI20}. Although these algorithms guarantee the \emph{privacy of the client} and the \emph{privacy of the server}, all of the homomorphic-encryption-based algorithms are limited to offline DFA execution and do not achieve \emph{online monitoring}.
We note that the extension of~\cite{DBLP:journals/joc/ChillottiGGI20} for online DFA execution is one of our technical contributions.

In~\cite{DBLP:conf/emsoft/Abbas19}, the authors propose an LTL runtime verification algorithm without revealing the monitored data to the server. They propose both offline and online algorithms to run a DFA converted from a safety LTL formula. The main issue with their online algorithm is that the DFA running on the server must be revealed to the client, and the goal of \emph{privacy of the server} is not satisfied.

\emph{Oblivious DFA evaluation (ODFA)}~\cite{DBLP:conf/ccs/Troncoso-PastorizaKC07,DBLP:conf/dbsec/Frikken09,DBLP:conf/dbsec/BlantonA10,DBLP:conf/wpes/SasakawaHdATS14,DBLP:conf/ctrsa/MohasselNSS12,DBLP:conf/pkc/GennaroHS10} is a technique to run a DFA on a server while keeping the DFA secret to the server and the monitored data secret to the client. Although the structure of the DFA is not revealed to the client, the client has to know the number of the states. Consequently, the goal 
\emph{privacy of the server} is satisfied \emph{only partially}. 
Moreover, to the best of our knowledge, none of the ODFA-based algorithms support online DFA execution. Therefore, the goal \emph{online monitoring} is not satisfied.

\noindent{\textit{Organization.}}
The rest of the paper is organized as follows:
In \cref{sec:preliminaries}, we overview LTL monitoring (\cref{subsec:ltl-mon}),
the FHE scheme we use (\cref{subsec:tfhe}), and
the leveled homomorphic offline algorithm (\cref{subsec:leveled}).
Then, in \cref{sec:proposed-alg},
we explain our fully homomorphic offline algorithm (\FHOffline) and
two online algorithms (\ReverseStream and \BlockStream).
We describe the proposed protocol for oblivious online LTL monitoring 
in \cref{sec:oblivious-ltl-mon}.
After we discuss our experimental results in \cref{sec:experiment},
we conclude our paper in \cref{sec:conclusion}.

\newcommand{\CiteTroncoso}{\cite{DBLP:conf/ccs/Troncoso-PastorizaKC07}}
\newcommand{\CiteFrikken}{\cite{DBLP:conf/dbsec/Frikken09}}
\newcommand{\CiteBlanton}{\cite{DBLP:conf/dbsec/BlantonA10}}
\newcommand{\CiteSasakawa}{\cite{DBLP:conf/wpes/SasakawaHdATS14}}
\newcommand{\CiteMohassel}{\cite{DBLP:conf/ctrsa/MohasselNSS12}}
\newcommand{\CiteChillotti}{\cite{DBLP:journals/joc/ChillottiGGI20}}
\newcommand{\CiteAbbas}{\cite{DBLP:conf/emsoft/Abbas19}}
\newcommand{\CiteGennrao}{\cite{DBLP:conf/pkc/GennaroHS10}}
\newcommand{\CiteIshai}{\cite{DBLP:conf/tcc/IshaiP07}}
\newcommand{\Sgm}{|\Sigma|}
\newcommand{\lgSgm}{\log |\Sigma|}
\newcommand{\cmark}{{\color{green}\ding{51}}}%
\newcommand{\xmark}{{\color{red}\ding{55}}}%
\newcommand{\qmark}{\textbf{?}}

\begin{table}[tb]\centering
	\caption{Related work on DFA execution with \emph{privacy of the client}.
	}\label{tbl:summary-related-work}
	\scriptsize
	\begin{tabular}{l|c|c|c|c|c|c|c|c|c||c}\toprule
		Work                      & \CiteTroncoso & \CiteFrikken & \CiteBlanton & \CiteSasakawa & \CiteMohassel & \CiteGennrao & \CiteIshai & \CiteChillotti & \CiteAbbas & Ours      \\\midrule
		Support online monitoring& \xmark        & \xmark       & \xmark       & \xmark        & \xmark        & \xmark       & \xmark     & \xmark         & \cmark     & \cmark\\
		Private the client's monitored data             & \cmark        & \cmark       & \cmark       & \cmark        & \cmark        & \cmark       & \cmark     & \cmark         & \cmark     & \cmark    \\
		Private DFA, except for its number of the states  & \cmark        & \cmark       & \cmark       & \cmark        & \cmark        & \cmark       & \cmark     & \cmark         & \xmark     & \cmark    \\
		Private DFA's number of the states             & \xmark        & \xmark       & \xmark       & \xmark        & \xmark        & \xmark       & \cmark     & \cmark         & \xmark     & \cmark    \\
		Performance report        & \xmark        & \cmark       & \xmark       & \cmark        & \cmark        & \xmark       & \xmark     & \xmark         & \xmark     & \cmark    \\
		\bottomrule
	\end{tabular}%
\end{table}

\section{Preliminaries}\label{sec:preliminaries}

\noindent{\textit{Notations.}}
We denote the set of all nonnegative integers by $\N$, the set of all positive integers by $\N^+$, and the set $\{0,1\}$ by $\B$.
Let $X$ be a set.
We write $2^X$ for the powerset of $X$.
We write $X^*$ for the set of finite sequences of $X$ elements and $X^\omega$ for the set of infinite sequences of $X$ elements.
For $u \in X^\omega$, we write $u_i \in X$ for the $i$-th element ($0$-based) of $u$, $u_{i:j} \in X^*$ for the subsequence $u_i, u_{i+1}, \dots, u_{j}$ of $u$, and $u_{i:} \in X^\omega$ for the suffix of $u$ starting from its $i$-th element.
For $u \in X^*$ and $v \in X^* \cup X^\omega$, we write $u \cdot v$ for the concatenation of $u$ and $v$.

\noindent{\textit{DFA.}}
A deterministic finite automaton (DFA) is a 5-tuple $(Q, \Sigma, \delta, q_0, F)$,
where $Q$ is a finite set of states, $\Sigma$ is a finite set of alphabet,
$\delta\colon Q\times\Sigma\to Q$
is a transition function,
$q_0\in Q$ is an initial state, and $F\subseteq Q$
is a set of final states.
If the alphabet of a DFA is $\mathbb{B}$, we call it a \emph{binary} DFA.
For a state $q\in Q$ and a word $w=\sigma_1\sigma_2\dots\sigma_{n}$
we define $\delta(q,w)\coloneqq\delta(\dots\delta(\delta(q,\sigma_1),\sigma_2),\dots,\sigma_{n})$.
For a DFA $M$ and a word $w$, we write $M(w)\coloneqq 1$ if $M$ accepts $w$;
otherwise, $M(w) \coloneqq 0$.
We also abuse the above notations for nondeterministic finite automata (NFAs).

\subsection{LTL}\label{subsec:ltl-mon}

We use \emph{linear temporal logic (LTL)}~\cite{DBLP:conf/focs/Pnueli77} to specify the monitored properties.
The following BNF defines the syntax of LTL formulae:
$
  \phi,\psi ::= \top
  \mid p
  \mid \lnot\phi
  \mid \phi\land\psi
  \mid \mathsf{X}\phi
  \mid \phi\mathsf{U}\psi
$,
where 
$\phi$ and $\psi$ range over LTL formulae and
$p$ ranges over a set $\AP$ of atomic propositions.

An LTL formula asserts a property of $u \in (2^\AP)^\omega$.
The sequence $u$ expresses an execution trace of a system; $u_i$ is the set of the atomic propositions satisfied at the $i$-th time step.
Intuitively, $\top$ represents an always-true proposition;
$p$ asserts that $u_0$ contains $p$, and hence $p$ holds at the $0$-th step in $u$;
$\lnot\phi$ is the negation of $\phi$; and
$\phi \land \psi$ is the conjunction of $\phi$ and $\psi$.
The temporal proposition $\mathsf{X} \phi$ expresses that $\phi$ holds from the next step (i.e., $u_{1:}$);
$\phi \mathsf{U} \psi$ expresses that $\psi$ holds eventually and $\phi$ continues to hold until then.
We write $\bot$ for $\neg \top$;
$\phi \lor \psi$ for $\neg (\neg \phi \land \neg \psi)$; 
$\phi \implies \psi$ for $\neg\phi \lor \psi$;
$\mathsf{F}\phi$ for $\top \mathsf{U} \phi$;
$\mathsf{G}\phi$ for $\neg (\mathsf{F} \neg\phi)$;
$\mathsf{G}_{[n,m]}\phi$ for
$\overbrace{\X\dots\X}^{n\text{ occurrences of }\X}(\phi\land\overbrace{\X(\phi\land\X(\dots\land\X\phi))}^{(m-n)\text{ occ. of }\X})$; and
$\mathsf{F}_{[n,m]}\phi$ for
$\overbrace{\X\dots\X}^{n\text{ occ. of }\X}(\phi\lor\overbrace{\X(\phi\lor\X(\dots\lor\X\phi))}^{(m-n)\text{ occ. of }\X})$.

We formally define the semantics of LTL below.
Let $u \in (2^{\AP})^{\omega}$, $i\in\mathbb{N}$, and $\phi$ be an LTL formula.
We define the relation $u,i \models \phi$ as the least relation that satisfies the following:
\begin{gather*}
 u,i \models \top \qquad u, i \models p \defiff p \in u(i) \qquad u,i\models\lnot\phi \defiff u,i\not\models\phi \\
 u,i\models\phi\land\psi \defiff u,i\models\phi\text{ and }u,i\models\psi \qquad\quad
 u,i\models\mathsf{X}\phi \defiff u,i+1\models\phi\\
 u,i\models \phi \mathsf{U} \psi \defiff \text{
 \begin{tabular}{l}
  there exists  $j\ge i$  such that  $u,j\models\psi$ and, \\
  for any $k$, $i\le k\le j \implies u,k\models\phi$.
 \end{tabular}}
\end{gather*}
We write $u \models \phi$ for $u,0\models \phi$ and say $u$ \emph{satisfies} $\phi$.

In this paper,
we focus on \emph{safety}~\cite{DBLP:journals/fmsd/KupfermanV01} (i.e., nothing bad happens) fragment of LTL properties.
A finite sequence $w \in (2^{\AP})^{*}$ is a \emph{bad prefix} for an LTL formula $\phi$ if $w\cdot v\not\models \phi$ holds for any $v\in (2^\AP)^\omega$. %
For any bad prefix $w$, we cannot extend $w$ to an infinite word that satisfies $\phi$.
An LTL formula $\phi$ is a
\emph{safety}
LTL formula if
for any $w\in (2^\AP)^\omega$ satisfying $w\not\models\phi$, $w$ has a bad prefix for $\phi$.

A \emph{safety monitor} (or simply a \emph{monitor}) is a procedure that takes $w \in (2^\AP)^\omega$ and a safety LTL formula $\phi$ and generates an alert if $w \not\models \phi$.
From the definition of safety LTL, it suffices for a monitor to detect a bad prefix of $\phi$.
It is known that, for any safety LTL formula $\phi$, we can construct a DFA $M_{\phi}$ recognizing the set of the bad prefixes of $\phi$~\cite{DBLP:journals/fmsd/TabakovRV12}, which can be used as a monitor.

\begin{table}[tb]
    \centering
    \caption{Summary of TFHE ciphertexts, where $N$ is a parameter of TFHE.}
    \label{tbl:summary-of-tfhe-ciphertexts}
    \scriptsize
    \begin{tabular}{c|c|c|c}\toprule
         Ciphertext Kind & Notation in this paper& Plaintext Message & Conversion from \TRLWE \\\midrule
        \TLWE  & $c$          & a Boolean value $b\in\B$ & \SampleExtract{} (fast) \\
        \TRLWE & $\mathbf{c}$ & a Boolean vector $v \in \B^N$ & ------------ \\
        \TRGSW & $d$          & a Boolean value $b \in \B$ & \begin{tabular}{@{}c@{}}
          \SampleExtract{} and     \\
          \CircuitBootstrapping{} (slow)
        \end{tabular} \\\bottomrule
    \end{tabular}
\end{table}

\subsection{Torus Fully Homomorphic Encryption}\label{subsec:tfhe}

Homomorphic encryption (HE) is a form of encryption that enables us to apply operations to encrypted values \emph{without decrypting them}.
In particular, a type of HE, called Fully HE (FHE), allows us to evaluate arbitrary functions over encrypted data~\cite{DBLP:journals/iacr/FanV12,DBLP:conf/stoc/Gentry09,DBLP:conf/crypto/GentrySW13,DBLP:conf/innovations/BrakerskiGV12}.
We use an instance of FHE called TFHE~\cite{DBLP:journals/joc/ChillottiGGI20} in this work.
We briefly summarize TFHE below; see~\cite{DBLP:journals/joc/ChillottiGGI20} for a detailed exposition.

We are concerned with the following two-party secure computation, where the involved parties are a client (called Alice) and a server (called Bob): %
\begin{oneenumeration}
 \item Alice generates the keys used during computation;
 \item Alice encrypts her plaintext messages into ciphertexts with her keys;
 \item Alice sends the ciphertexts to Bob;
 \item Bob conducts computation over the received ciphertexts and obtains the encrypted result \emph{without decryption};
 \item Bob sends the encrypted results to Alice;
 \item Alice decrypts the received results and obtains the results in plaintext.
\end{oneenumeration}

\subsubsection{Keys.}\label{subsec:keys}
There are three types of keys in TFHE: \emph{secret key} $\SK$, \emph{public key} $\PK$, and \emph{bootstrapping key} $\BK$.
All of them are generated by Alice.
$\PK$ is used to encrypt plaintext messages into ciphertexts, and
$\SK$ is used to decrypt ciphertexts into plaintexts.
Alice keeps $\SK$ private, i.e., the key is known only to herself but not to Bob.
In contrast, $\PK$ is public and also known to Bob.
$\BK$ is generated from $\SK$ and can be safely shared
with Bob without revealing $\SK$.
$\BK$ allows Bob to evaluate the homomorphic operations (defined later) %
over the ciphertext.

\subsubsection{Ciphertexts.}
Using the public key, Alice can generate three kinds of ciphertexts (\cref{tbl:summary-of-tfhe-ciphertexts}):
\TLWE (Torus Learning With Errors),
\TRLWE (Torus Ring Learning With Errors),
and \TRGSW (Torus Ring Gentry-Sahai-Waters).
Homomorphic operations provided by TFHE %
are defined over each of the specific ciphertexts.
We note that different ciphertexts have different data structures, and their conversion can be time-consuming.
\cref{tbl:summary-of-tfhe-ciphertexts} shows one such example.

In TFHE, different types of ciphertexts represent different plaintext messages.
A \TLWE ciphertext represents a Boolean value.
In contrast, \TRLWE represents a vector of Boolean values of length $N$,
where $N$ is a TFHE parameter.
We can regard a \TRLWE ciphertext as a vector of \TLWE ciphertexts,
and the conversion between a \TRLWE ciphertext and a \TLWE one
is relatively easy. %
A \TRGSW ciphertext also represents a Boolean value, but its data structure is quite different from \TLWE, and
the conversion from \TLWE to
\TRGSW %
is slow. %

TFHE provides different encryption and decryption functions for each type of ciphertext.
We write $\Enc(x)$ for a ciphertext of a plaintext $x$; $\Dec(c)$ for the plaintext message for the ciphertext $c$.
We abuse these notations for all three types of ciphertexts.

Besides, TFHE supports \emph{trivial samples} of \TRLWE.
A trivial sample of \TRLWE has the same data structure as a \TRLWE ciphertext
but is \emph{not} encrypted, i.e., anyone can tell the plaintext message
represented by the trivial sample.
We denote by $\Trivial(n)$ a trivial sample of \TRLWE
whose plaintext message is $(b_1,b_2,\dots,b_{N})$, where each $b_{i}$ is the
$i$-th bit in the binary representation of $n$\begin{ExtendedVersion}, i.e., $n = \sum_{i=1}^{N} b_i 2^{i-1}$\end{ExtendedVersion}.

\subsubsection{Homomorphic Operations.}\label{subsec:functions_over_ciphertexts}
TFHE provides \emph{homomorphic operations}, i.e.,
operations over ciphertexts without decryption.
Among the operators supported by TFHE~\cite{DBLP:journals/joc/ChillottiGGI20},
we use the following ones.

\begin{description}
\item[$\CMux(d,\mathbf{c_{\text{true}}},\mathbf{c_{\text{false}}})\text{\normalfont{} : $\TRGSW\times\TRLWE\times\TRLWE\to\TRLWE$}$]\mbox{}\\
    Given a \TRGSW ciphertext $d$ and \TRLWE ciphertexts $\mathbf{c}_{\text{true}}, \mathbf{c}_{\text{false}}$,
    \CMux{} outputs a \TRLWE ciphertext $\mathbf{c}_{\text{result}}$ such that
    $\Dec(\mathbf{c}_{\text{result}})=\Dec(\mathbf{c}_\text{true})$ if $\Dec(d)=1$, and otherwise,
    $\Dec(\mathbf{c}_{\text{result}})=\Dec(\mathbf{c}_\text{false})$.

\item[$\LookUp(\{\mathbf{c}_i\}_{i=1}^{2^n}, \{d_i\}_{i=1}^n)\text{\normalfont{} : $(\TRLWE)^{2^n}\times(\TRGSW)^n\to\TRLWE$}$]\mbox{}\\
    Given \TRLWE ciphertexts $\mathbf{c}_1,\mathbf{c}_2,\dots,\mathbf{c}_{2^n}$ and \TRGSW ciphertexts $d_1,d_2,\dots,\allowbreak d_{n}$,
    \LookUp{} outputs a \TRLWE ciphertext $\mathbf{c}$ such that $\Dec(\mathbf{c}) = \Dec(\mathbf{c}_k)$ and $k = \sum_{i=1}^{n} 2^{i-1}\times\Dec(d_i)$.

\item[$\SampleExtract(k, \mathbf{c})\text{\normalfont{} : $\mathbb{N}\times\TRLWE\to\TLWE$}$]\mbox{}\\
    Let $\Dec(\mathbf{c}) = (b_1,b_2,\dots,b_{N})$.
    Given $k < N$ and a \TRLWE ciphertext $\mathbf{c}$,
    \SampleExtract outputs a \TLWE ciphertext $c$ where $\Dec(c)=b_{k+1}$.
\end{description}

Intuitively, \CMux{} can be regarded as a multiplexer over \TRLWE ciphertexts with \TRGSW selector input.
The operation \LookUp{} regards $\mathbf{c}_1,\mathbf{c}_2,\dots,\mathbf{c}_{2^n}$
as encrypted entries composing a LookUp Table (LUT) of depth $n$ and $d_1,d_2,\dots,d_{n}$
as inputs to the LUT. Its output is the entry selected by the LUT.
\LookUp{} is constructed by $2^n-1$ \CMux{} arranged in a tree of depth $n$.
\SampleExtract{} outputs the $k$-th element of $\mathbf{c}$ as \TLWE.
Notice that all these operations work over ciphertexts without decrypting them.

\noindent\textbf{Noise and Operations for Noise Reduction.}
In generating a TFHE ciphertext, we ensure its security by adding some random numbers called \emph{noise}.
An application of a TFHE operation adds noise to its output ciphertext;
if the noise in a ciphertext becomes too large, the TFHE ciphertext cannot be correctly decrypted.
There is a special type of operation called \emph{bootstrapping}\footnote{%
Note that bootstrapping here has nothing to do with bootstrapping in statistics.%
}~\cite{DBLP:conf/stoc/Gentry09},
which reduces the noise of a TFHE ciphertext.

\begin{description}
\item[$\Bootstrapping_\BK(c)\text{\normalfont : $\TLWE \to \TRLWE$}$]\mbox{}\\
    Given a bootstrapping key $\BK$ and a \TLWE ciphertext $c$,
    \Bootstrapping outputs a \TRLWE ciphertext $\mathbf{c}$ where $\Dec(\mathbf{c}) = (b_1, b_2, \dots, b_{N})$ and $b_1 = \Dec(c)$.
    Moreover, the noise of $\mathbf{c}$ becomes a constant that is determined by the parameters of TFHE and is independent of $c$.
    
\item[$\CircuitBootstrapping_\BK(c)\text{\normalfont : $\TLWE \to \TRGSW$}$]\mbox{}\\
    Given a bootstrapping key $\BK$ and a \TLWE ciphertext $c$,
    \CircuitBootstrapping outputs a \TRGSW ciphertext $d$ where $\Dec(d) = \Dec(c)$.
    The noise of $d$ becomes a constant that is determined by the parameters of TFHE and is independent of $c$.
\end{description}

These bootstrapping operations are used to keep the noise of a TFHE ciphertext small enough to be correctly decrypted.
\Bootstrapping and \CircuitBootstrapping are
almost two and three orders of magnitude slower than \CMux{},
respectively~\cite{DBLP:journals/joc/ChillottiGGI20}.

\subsubsection{Parameters for TFHE.}
There are many parameters for TFHE, such as
the length $N$ of the message of a \TRLWE ciphertext and
the standard deviation of the probability distribution from which a noise is taken.
Certain properties of TFHE depend on these parameters.
These properties include the security level of TFHE,
the number of TFHE operations that can be applied without bootstrapping ensuring correct decryption,
and the time and the space complexity of each operation.
The complete list of TFHE parameters is presented in \crefAppendix{TFHE_parameters}.

We remark that we need to determine the TFHE parameters
\emph{before} performing any TFHE operation.
Therefore, we need to know
the number of applications of homomorphic operations without bootstrapping \emph{in advance}, i.e.,
the homomorphic circuit depth must be determined \emph{a priori}.

\subsection{Leveled Homomorphic Offline Algorithm}%
\label{subsec:leveled}

\begin{algorithm}[tb]\scriptsize
\Input{A binary DFA $M=(Q, \Sigma=\B, \delta, q_0, F)$ and \TRGSW monitored ciphertexts $d_1, d_2,\dots, d_{n}$}
\Output{A \TLWE ciphertext $c$ satisfying $\Dec(c) = M(\Dec(d_1)\Dec(d_2)\dots\Dec(d_{n}))$}
\For{$q \in Q$}{
    \label{line:unf-off:init0}
    $\mathbf{c}_{n,q} \gets q \in F\text { ? }\textsc{Trivial}(1)\text{ : }\textsc{Trivial}(0)$
    \tcp*{Initialize each $\mathbf{c}_{n,q}$}
    \label{line:unf-off:init1}
}
\For{$i=n,n-1,\dots,1$}{
    \label{line:unf-off:for00}
    \For{$q \in Q$ such that $q$ is reachable from $q_0$ by $(i-1)$ transitions} {
        $\mathbf{c}_{i-1,q} \gets \CMux{}(d_i, \mathbf{c}_{i,\delta(q, 1)}, \mathbf{c}_{i,\delta(q, 0)})$
    }
    \label{line:unf-off:for01}
}
$c \gets \SampleExtract{}(0,\mathbf{c}_{0,q_0})$\;\label{line:unf-off:sampleextract}
\Return $c$
\caption{The leveled homomorphic offline algorithm~\cite{DBLP:journals/joc/ChillottiGGI20}.}\label{alg:unfixed-parameter-offline}
\end{algorithm}

\begin{sloppypar}
Chillotti et al.~\cite{DBLP:journals/joc/ChillottiGGI20} proposed an \emph{offline} algorithm to evaluate a DFA over TFHE ciphertexts (\cref{alg:unfixed-parameter-offline}).
Given a DFA $M$ and \TRGSW ciphertexts $d_1, d_2, \dots, d_{n}$,
\cref{alg:unfixed-parameter-offline} returns a \TLWE ciphertext $c$ satisfying
$\Dec(c) = M(\Dec(d_1)\Dec(d_2)\dots\Dec(d_{n}))$.
For simplicity, for a state $q$ of $M$,
we write $M^i(q)$ for $M(q,\Dec(d_i)\Dec(d_{i+1})\dots\Dec(d_n))$.
\end{sloppypar}

In \cref{alg:unfixed-parameter-offline},
we use a \TRLWE ciphertext $\mathbf{c}_{i,q}$ whose first element represents $M^{i+1}(q)$, i.e., whether we reach a final state by reading $\Dec(d_{i+1})\Dec(d_{i+2})\dots\Dec(d_n)$ from $q$.
We abuse this notation for $i=n$, i.e., the first element of $\mathbf{c}_{n,q}$ represents if $q \in F$.
In \cref{line:unf-off:init0,line:unf-off:init1}, we initialize $\mathbf{c}_{n,q}$;
For each $q\in Q$, we let $\mathbf{c}_{n,q}$ be $\Trivial(1)$ if $q \in F$; otherwise, we let $\mathbf{c}_{n,q}$ be $\Trivial(0)$.
In \crefrange{line:unf-off:for00}{line:unf-off:for01},
we construct $\mathbf{c}_{i-1,q}$ inductively by feeding each monitored ciphertext $d_i$ to \CMux{} from tail to head.
Here, $\mathbf{c}_{i-1,q}$ represents $M^i(q)$ because of
$M^{i}(q) = M^{i+1}(\delta(q,\Dec(d_i)))$.
We note that for the efficiency, we only construct $\mathbf{c}_{i-1,q}$ for the states reachable from $q_0$ by $i-1$ transitions.
In \cref{line:unf-off:sampleextract}, we extract the first element of $\mathbf{c}_{0,q_0}$,
which represents $M^1(q_0)$, i.e., $M(\Dec(d_1)\Dec(d_2)\dots\Dec(d_n))$.

\begin{theorem}[{Correctness~\cite[Thm.~5.4]{DBLP:journals/joc/ChillottiGGI20}}]\label{thm:leveled-offline}
    Given a binary DFA $M$ and \TRGSW ciphertexts $d_1,d_2,\dots,d_n$,
    if $c$ in \cref{alg:unfixed-parameter-offline} can be correctly decrypted,
    \cref{alg:unfixed-parameter-offline} outputs $c$ satisfying $\Dec(c)=M(\Dec(d_1)\Dec(d_2)\dots\Dec(d_n))$.
    \qed{}
\end{theorem}

\noindent{\textbf{Complexity Analysis.}}
The time complexity of \cref{alg:unfixed-parameter-offline} is determined by the number of applications of
\CMux{}, which is $O(n|Q|)$.
Its space complexity is $O(|Q|)$ because we can use two sets of $|Q|$ \TRLWE ciphertexts alternately
for $\mathbf{c}_{2j-1,q}$ and $\mathbf{c}_{2j,q}$ (for $j\in\N^+$).

\noindent{\textbf{Shortcomings of \cref{alg:unfixed-parameter-offline}.}}
We cannot use \cref{alg:unfixed-parameter-offline} under an \emph{online} setting due to two reasons.
Firstly, \cref{alg:unfixed-parameter-offline} is a \emph{leveled} homomorphic algorithm,
i.e., the maximum length of the ciphertexts that \cref{alg:unfixed-parameter-offline} can handle is determined by TFHE parameters.
This is because \cref{alg:unfixed-parameter-offline} does not use \Bootstrapping, and if the monitored ciphertexts are too long, the result $c$ cannot be correctly decrypted due to the noise.
This is critical in an online setting because
we do not know the length $n$ of the monitored ciphertexts in advance, and
we cannot determine such parameters appropriately.

Secondly, \cref{alg:unfixed-parameter-offline} consumes the monitored ciphertext from back to front, 
i.e., the last ciphertext $d_{n}$ is used in the beginning, and $d_{1}$ is used in the end.
Thus, we cannot start \cref{alg:unfixed-parameter-offline} before the last input is given.

\section{Online Algorithms for Running DFA Obliviously}\label{sec:proposed-alg}
In this section, we propose two online algorithms that run a DFA obliviously.
As a preparation for these online algorithms, we also introduce a fully homomorphic offline algorithm based on \cref{alg:unfixed-parameter-offline}.

\subsection{Preparation: Fully Homomorphic Offline Algorithm (\FHOffline)}
\begin{algorithm}[tb]\scriptsize
\Input{A binary DFA $M=(Q, \Sigma=\B, \delta, q_0, F)$, \TRGSW monitored ciphertexts $d_1, d_2, \dots, d_{n}$, a bootstrapping key $\BK$, and $\BootstrappingInterval\in\mathbb{N}^+$}
\Output{A \TLWE ciphertext $c$ satisfying $\Dec(c) = M(\Dec(d_1)\Dec(d_2)\dots\Dec(d_{n}))$}
\For{$q \in Q$}{
    $\mathbf{c}_{n,q} \gets q \in F\text { ? }\textsc{Trivial}(1)\text{ : }\textsc{Trivial}(0)$\;
}
\For{$i=n,n-1,\dots,1$}{\label{line:offline:for0}
    \For{$q \in Q$ such that $q$ is reachable from $q_0$ by $(i-1)$ transitions} {
        $\mathbf{c}_{i-1,q} \gets \CMux{}(d_i, \mathbf{c}_{i,\delta(q, 1)}, \mathbf{c}_{i,\delta(q, 0)})$
    }
    \color{red}{%
    \If{$(n-i+1)\mod \BootstrappingInterval = 0$}{
        \label{line:offline:bs0}
        \For{$q\in Q$ such that reachable from $q_0$ by $(i-1)$ transitions}{
            $c_{i-1,q} \gets \SampleExtract(0,\mathbf{c}_{i-1,q})$\;
            $\mathbf{c}_{i-1,q} \gets \Bootstrapping_\BK(c_{i-1,q})$
            \label{line:offline:bs1}\label{line:offline:for1}
        }
    }}
}
$c \gets \SampleExtract{}(0,\mathbf{c}_{0,q_0})$\;
\Return $c$

\caption{Our fully homomorphic offline algorithm (\FHOffline).}\label{alg:fixed-parameter-offline}\label{alg:offline}
\end{algorithm}

As preparation for introducing an algorithm that can run a DFA under an online setting, we enhance \cref{alg:unfixed-parameter-offline} so that we can monitor a sequence of ciphertexts whose length is unknown \emph{a priori}.
\cref{alg:fixed-parameter-offline} shows our \emph{fully homomorphic} offline algorithm (\FHOffline), which does not require TFHE parameters to depend on the length of the monitored ciphertexts.
The key difference lies in \crefrange{line:offline:bs0}{line:offline:bs1} (the red lines) of \cref{alg:fixed-parameter-offline}.
Here, for every $\BootstrappingInterval$ consumption of the monitored ciphertexts, we reduce the noise by applying \Bootstrapping{} to the ciphertext $\mathbf{c}_{i,j}$ representing a state of the DFA.
Since the amount of the noise accumulated in $\mathbf{c}_{i,j}$ is determined only by the number of the processed ciphertexts,
we can keep the noise levels of $\mathbf{c}_{i,j}$ low and ensure that the monitoring result $c$ is correctly decrypted.
Therefore, by using \cref{alg:fixed-parameter-offline},
we can monitor an arbitrarily long sequence of ciphertexts as long as the interval $\BootstrappingInterval$ is properly chosen according to the TFHE parameters.
We note that we still cannot use \cref{alg:fixed-parameter-offline} for online monitoring because it consumes the monitored ciphertexts from back to front.

\subsection{Online Algorithm 1: \ReverseStream}\label{subsec:reversed}
\begin{algorithm}[tb]\scriptsize
\Input{A binary DFA $M$, \TRGSW monitored ciphertexts $d_1, d_2, d_3, \dots, d_n$, a bootstrapping key $\BK$, and $\BootstrappingInterval\in\mathbb{N}^+$}
\Output{For every $i\in \{1,2,\dots,n\}$, a \TLWE ciphertext $c_i$ satisfying $\Dec(c_i) = M(\Dec(d_1)\Dec(d_2)\dots\Dec(d_i))$}
let $\Rev{M}=(\Rev{Q},\B,\Rev{\delta},\Rev{q}_0,\Rev{F})$ be the minimum reversed DFA of $M$
\;\label{line:reversed:reverse}
\For{$\Rev{q}\in\Rev{Q}$}{
    $\mathbf{c}_{0,\Rev{q}} \gets \Rev{q}\in\Rev{F}\text{ ? }\Trivial(1)\text{ : }\Trivial(0)$\;
}
\For{$i = 1,2,\dots,n$}{\label{line:reversed:while0}
    \For{$\Rev{q}\in\Rev{Q}$}{\label{line:reversed:for0}
          $\mathbf{c}_{i,\Rev{q}} \gets \CMux{}(d_i, \mathbf{c}_{i-1,\Rev{\delta}(\Rev{q}, 1)}, \mathbf{c}_{i-1,\Rev{\delta}(\Rev{q}, 0)})$
    }
    \If{$i \mod \BootstrappingInterval = 0$}{
        \For{$\Rev{q}\in\Rev{Q}$}{\label{line:reversed:for00}
            $c_{i,\Rev{q}} \gets \SampleExtract(0,\mathbf{c}_{i,\Rev{q}})$\;\label{line:reversed:sample_extract}
            $\mathbf{c}_{i,\Rev{q}} \gets \Bootstrapping{}_\BK(c_{i,\Rev{q}})$\label{line:reversed:bootstrapping}\label{line:reversed:for01}
        }
    }
    $c_{i} \gets \SampleExtract{}(0,\mathbf{c}_{i,\Rev{q_0}})$\;
    \Outputs $c_i$\;\label{line:reversed:while1}
}
\caption{Our first online algorithm (\ReverseStream).}\label{alg:reversed}
\end{algorithm}

To run a DFA online, we modify \FHOffline so that the monitored ciphertexts are consumed from front to back. 
Our main idea is illustrated in \cref{fig:outline_algorithm:reversed}: we  \emph{reverse} the DFA $M$ beforehand and feed the ciphertexts $d_1, d_2,\dots,d_n$ to the reversed DFA $\Rev{M}$ serially from $d_1$ to $d_n$.

\cref{alg:reversed} shows the outline of our first online algorithm (\ReverseStream) based on the above idea.
\ReverseStream takes the same inputs as \FHOffline: 
a DFA $M$, \TRGSW ciphertexts $d_1, d_2, \dots, d_n$, a bootstrapping key $\BK$, and a positive integer $\BootstrappingInterval$ indicating the interval of bootstrapping.
In \cref{line:reversed:reverse},
we construct the minimum DFA $\Rev{M}$ that satisfies, for any $w=\sigma_1\sigma_2\dots\sigma_{k} \in \B^*$, 
we have $\Rev{M}(w) = M(w^R)$, where $w^R=\sigma_{k}\dots\sigma_1$.
We can construct such a DFA by reversing the transitions and by applying the powerset construction and the minimization algorithm\begin{ExtendedVersion}~\cite{DBLP:books/daglib/0000197}\end{ExtendedVersion}.

In the loop from \crefrange{line:reversed:while0}{line:reversed:while1},  the reversed DFA $\Rev{M}$ consumes each monitored ciphertext $d_i$, 
which corresponds to the loop from \crefrange{line:offline:for0}{line:offline:for1} in \cref{alg:fixed-parameter-offline}.
The main difference lies in \cref{line:reversed:for0,line:reversed:for00}: \cref{alg:reversed} applies \CMux and \Bootstrapping to all the states of $\Rev{M}$,
while \cref{alg:offline} only considers the states reachable from the initial state.
This is because in online monitoring, we monitor a stream of ciphertexts without knowing the number of the remaining ciphertexts, and all the states of the reversed DFA $\Rev{M}$ are potentially reachable from the initial state $\Rev{q_0}$ by the reversed remaining ciphertexts $d_n,d_{n-1},\dots,d_{i+1}$ because of the minimality of $\Rev{M}$.

\begin{theorem}\label{thm:correctness-reversed}
    Given a binary DFA $M$, \TRGSW ciphertexts $d_1,d_2,\dots,d_n$,
    a bootstrapping key $\BK$, and a positive integer $\BootstrappingInterval$,
    for each $i \in \{1,2,\dots,n\}$,
    if $c_i$ in \cref{alg:reversed} can be correctly decrypted,
    \cref{alg:reversed} outputs $c_i$ satisfying
    $\Dec(c_i) = M(\Dec(d_1)\Dec(d_2)\dots\Dec(d_{i}))$.

\end{theorem}
\noindent{\textit{Proof (sketch).}}
    \SampleExtract and \Bootstrapping in \cref{line:reversed:sample_extract,line:reversed:bootstrapping} do not change the decrypted value of $c_i$.
    Therefore, $\Dec(c_i)=\Rev{M}(\Dec(d_{i})\dots\Dec(d_1))$ for $i \in \{1,2,\dots, n\}$ by \cref{thm:leveled-offline}.
    As $\Rev{M}$ is the reversed DFA of $M$, we have
    \begin{math}
        \Dec(c_i) = \Rev{M}(\Dec(d_{i})\dots\Dec(d_1))
                  = M(\Dec(d_1)\dots\Dec(d_{i})).
    \end{math}
\qed{}

\subsection{Online Algorithm 2: \BlockStream}\label{subsec:bbs}

\newcommand{\BlockSize}{B}
\begin{algorithm}[tb]\scriptsize
\Input{A binary DFA $M=(Q, \Sigma=\B, \delta, q_0, F)$, \TRGSW monitored ciphertexts $d_1, d_2, d_3, \dots, d_n$, a bootstrapping key $\BK$, and $\BlockSize\in\mathbb{N}^+$}
\Output{For every $i\in\mathbb{N}^+$ ($i\le\lfloor n/\BlockSize\rfloor$), a \TLWE ciphertext $c_i$ satisfying $\Dec(c_i) = M(\Dec(d_1)\Dec(d_2)\dots\Dec(d_{i\times\BlockSize}))$}

$S_1 \gets \{q_0\}$ \tcp*{$S_i$: the states reachable by $(i-1) \times \BlockSize$ transitions.}\label{line:bbs:init}
\For{$i = 1,2,\dots, \lfloor n/\BlockSize\rfloor$}{\label{line:bbs:while0}
    $S_{i+1} \gets \{q\in Q \mid \exists s_i \in S_i.\, \text{$q$ is reachable from $s_i$ by $\BlockSize$ transitions} \}$\;\label{line:bbs:calc-Snext}
    \tcp*{We denote $S_{i+1} = \{s^{i+1}_1, s^{i+1}_2,\dots, s^{i+1}_{|S_{i+1}|}\}$}
    \For{$q\in Q$}{\label{line:bbs:offline0}
        \If{$q\in S_{i+1}$}{
            $j \gets \text{the index of $S_{i+1}$ such that $q = s_j^{i+1}$}$\;
            $\mathbf{c}^{T_i}_{\BlockSize,q} \gets \Trivial((j-1)\times 2 + (q\in F\text{ ? }1\text{ : }0))$\label{line:bbs:trivial}
        }
    }
    \For{$k=\BlockSize,\BlockSize-1,\dots,1$}{
        \For{$q\in Q$ such that $q$ is reachable from a state in $S_i$ by $(k-1)$ transitions}{
            $\mathbf{c}^{T_i}_{k-1,q} \gets \CMux(d_{(i-1)\BlockSize+k},\mathbf{c}^{T_i}_{k,\delta(q,1)}, \mathbf{c}^{T_i}_{k,\delta(q,0)})$\label{line:bbs:offline1}\label{line:bbs:cmux}
        }
    }
    \eIf{$|S_i|=1$}{\label{line:bbs:select0}
        $\mathbf{c}^\text{cur}_{i+1} \gets \mathbf{c}^{T_i}_{0,q}$ where $S_i = \{q\}$\;\label{line:bbs:assign-size-1}
    }{
        \For{$l=1,2,\dots,\lceil\log_2(|S_i|)\rceil$}{
            $c_l \gets \SampleExtract(l,\mathbf{c}^\text{cur}_{i})$\;\label{line:bbs:sample_extract}
            $d'_l \gets \CircuitBootstrapping{}_{\BK}(c_l)$\;\label{line:bbs:circuit_bootstrapping}
        }
        $\mathbf{c}^\text{cur}_{i+1} \gets \LookUp{}(\{\mathbf{c}^{T_i}_{0,s_1^i},\mathbf{c}^{T_i}_{0,s_2^i},\dots \mathbf{c}^{T_i}_{0,s_{|S_i|}^i}\},\{d'_1, \dots, d'_{\lceil\log_2(|S_i|)\rceil}\})$\;\label{line:bbs:select1}\label{line:bbs:lookup}
    }
    $c_i \gets \SampleExtract{}(0,\mathbf{c}^\text{cur}_{i+1})$\;\label{line:bbs:output0}
    \Outputs $c_i$\;\label{line:bbs:output1}\label{line:bbs:while1}
}

\caption{Our second online algorithm (\BlockStream).}\label{alg:bbs}
\end{algorithm}

A problem of \ReverseStream{} is that
the number of the states of the reversed DFA can explode
exponentially due to powerset construction (see \cref{subsec:complexity_analysis} for the details).
Another idea of an online algorithm without reversing a DFA is illustrated in \cref{fig:outline_algorithm:bbs}:
we split the monitored ciphertexts into \emph{blocks} of fixed size $\BlockSize$ and process each block in the same way as \cref{alg:fixed-parameter-offline}.
Intuitively, 
for each block $d_{1 + (i - 1) \times \BlockSize}, d_{2 + (i - 1) \times \BlockSize},\dots,\allowbreak d_{\BlockSize + (i - 1) \times \BlockSize}$ of ciphertexts,
we compute the function $T_{i}\colon Q \to Q$ satisfying $
    T_{i}(q) = \delta(q, d_{1 + (i - 1) \times \BlockSize}, d_{2 + (i - 1) \times \BlockSize},\dots,d_{\BlockSize + (i - 1) \times \BlockSize})
$ by a variant of \FHOffline, and keep track of the current state of the DFA after reading  the current prefix $d_{1}, d_{2},\dots,d_{\BlockSize + (i - 1) \times \BlockSize}$.

\Cref{alg:bbs} shows the outline of our second online algorithm (\BlockStream) based on the above idea.
\Cref{alg:bbs} takes a DFA $M$,
\TRGSW ciphertexts $d_1,d_2,\dots,d_n$,
a bootstrapping key $\BK$, and an integer $\BlockSize$ representing the interval of output.
To simplify the presentation, we make the following assumptions, which are relaxed later:
\begin{oneenumeration}
 \item $\BlockSize$ is small, and a trivial \TRLWE sample can be correctly decrypted
after $\BlockSize$ applications of \CMux;
 \item the size $|Q|$ of the states of the DFA $M$ is smaller than or equal to $2^{N-1}$,
where $N$ is the length of \TRLWE.
\end{oneenumeration}

The main loop of the algorithm is sketched on \crefrange{line:bbs:while0}{line:bbs:while1}.
In each iteration, we consume the $i$-th block consisting of $\BlockSize$  ciphertexts, i.e., $d_{(i-1)\BlockSize+1}, \dots, d_{(i-1)\BlockSize + \BlockSize}$.
In \cref{line:bbs:calc-Snext}, we compute the set
$S_{i + 1} = \{s^{i+1}_1, s^{i+1}_2, \dots, s^{i+1}_{|S_\text{next}|}\}$ 
of the states reachable from $q_0$ by reading a word of length $i \times \BlockSize$.

In \crefrange{line:bbs:offline0}{line:bbs:offline1}, for each $q \in Q$, we construct a ciphertext representing $T_{i}(q)$ by feeding the current block to a variant of \FHOffline.
More precisely, we construct a ciphertext $\mathbf{c}^{T_i}_{0,q}$ 
representing the pair of the Boolean value showing if $T_i(q) \in F$
and the state $T_i(q) \in Q$.
The encoding of such a pair in a \TRLWE ciphertext is as follows: the first element shows if $T_i(q) \in F$
and the other elements are the binary representation of $j \in \mathbb{N}^+$, where $j$ is such that
$s^{i+1}_{j} = T_i(q)$.

In \crefrange{line:bbs:select0}{line:bbs:select1}, we construct the ciphertext $\mathbf{c}^\text{cur}_{i+1}$ representing the state of the DFA $M$ after reading the current prefix $d_{1}, d_{2},\dots,d_{\BlockSize + (i - 1) \times \BlockSize}$.
If $|S_i|=1$, since the unique element $q$ of $S_i$ is the only possible state before consuming the current block, the state after reading it is $T(q)$.
Therefore, we let $\mathbf{c}^\text{cur}_{i+1} = \mathbf{c}^{T_i}_{0,q}$.

Otherwise, we extract the ciphertext representing the state $q$ before consuming the current block, and 
let $\mathbf{c}^\text{cur}_{i+1} = \mathbf{c}^{T_i}_{0,q}$.
Since the $\mathbf{c}^\text{cur}_i$ (except for the first element) represents $q$ (see \cref{line:bbs:trivial}), 
we extract them by applying \SampleExtract{} (\cref{line:bbs:sample_extract}) and 
convert them to \TRGSW by applying \CircuitBootstrapping{} (\cref{line:bbs:circuit_bootstrapping}).
Then, we choose $\mathbf{c}^{T_i}_{0,q}$ by applying \LookUp{} and set it to $\mathbf{c}^\text{cur}_{i+1}$.

The output after consuming the current block, i.e., $M(\Dec(d_1)\Dec(d_2)\dots\allowbreak\Dec(d_{(i-1)\BlockSize+\BlockSize}))$ is stored in the first element of the \TRLWE ciphertext $\mathbf{c}^\text{cur}_{i+1}$.
It is extracted by applying \SampleExtract{} in~\cref{line:bbs:output0} and output in~\cref{line:bbs:output1}.

\begin{theorem}%
    \label{thm:correctness-bbs}
    Given a binary DFA $M$, \TRGSW ciphertexts $d_1,d_2,\dots,d_n$,
    a bootstrapping key $\BK$, and a positive integer $\BlockSize$,
    for each $i \in \{1,2,\dots, \lfloor n/\BlockSize\rfloor\}$,
    if $c_i$ in \cref{alg:bbs} can be correctly decrypted,
    \cref{alg:bbs} outputs a \TLWE ciphertext $c_i$ satisfying
    $\Dec(c_i) = M(\Dec(d_1)\Dec(d_2)\dots\Dec(d_{i \times \BlockSize}))$.
\end{theorem}
\begin{proof}[sketch]
    Let $q^i$ be $\delta(q_0, \Dec(d_1)\Dec(d_2)\dots\Dec(d_{i \times \BlockSize}))$.
    It suffices to show that, for each iteration $i$ in \cref{line:bbs:while0},
    $\Dec(\mathbf{c}_{i+1}^\text{cur})$ represents a pair of the Boolean value showing if $q^i\in F$
    and the state $q^i\in Q$ in the above encoding format.
    This is because $c_i$ represents the first element of $\mathbf{c}_{i+1}^\text{cur}$.
    \Cref{alg:bbs} selects $\mathbf{c}_{i+1}^\text{cur}$ from $\{\mathbf{c}_{0,q}^{T_i}\}_{q\in S_i}$ in \cref{line:bbs:assign-size-1} or \cref{line:bbs:lookup}.
    By using a slight variant of \cref{thm:leveled-offline} in \crefrange{line:bbs:select0}{line:bbs:select1},
    we can show that $\mathbf{c}_{0,q}^{T_i}$ represents if $T^i(q)\in F$ and the state $T^i(q)$.
    Therefore, the proof is completed by showing $\Dec(\mathbf{c}_{i+1}^\text{cur}) = \Dec(\mathbf{c}_{0,q^{i-1}}^{T_i})$.

    We prove $\Dec(\mathbf{c}_{i+1}^\text{cur}) = \Dec(\mathbf{c}_{0,q^{i-1}}^{T_i})$ by induction on $i$.
    If $i = 1$, $|S_i| = 1$ holds, and by $q^{i-1} \in S_{i}$,
    we have $\Dec(\mathbf{c}_{i+1}^\text{cur}) = \Dec(\mathbf{c}_{0,q^{i-1}}^{T_i})$.
    If $i > 1$ and $|S_i| = 1$,
    $\Dec(\mathbf{c}_{i+1}^\text{cur}) = \Dec(\mathbf{c}_{0,q^{i-1}}^{T_i})$ holds similarly.
    If $i > 1$ and $|S_i| > 1$, by induction hypothesis,
    $\Dec(\mathbf{c}_{i}^\text{cur})$ represents if $T_{i-1}(q^{i-2}) = q^{i-1}\in F$ and the state $q^{i-1}$.
    By construction in \cref{line:bbs:circuit_bootstrapping},
    $\Dec(d'_l)$ is equal to the $l$-th bit of $(j-1)$, where $j$ is such that $s_j^i = q^{i-1}$.
    Therefore, the result of the application of \LookUp in \cref{line:bbs:select1}
    is equivalent to $\mathbf{c}_{0,s^{i}_{j}}^{T_i} (= \mathbf{c}_{0,q^{i-1}}^{T_i})$, and 
    we have $\Dec(\mathbf{c}_{i+1}^\text{cur}) = \Dec(\mathbf{c}_{0,q^{i-1}}^{T_i})$.
    \qed{}
\end{proof}

We note that \BlockStream generates output
for every $B$ monitored ciphertexts
while \ReverseStream generates output for every monitored ciphertext.

We also remark that when $B=1$, \BlockStream{} consumes every monitored ciphertext from front to back.
However, such a setting is slow due to a huge number of \CircuitBootstrapping{} operations, as pointed out in \cref{subsec:complexity_analysis}.

\noindent{\textbf{Relaxations of the Assumptions.}}
When $\BlockSize$ is too large, $\mathbf{c}_{0,q}^{T_i}$ may not be correctly decrypted.
We can relax this restriction by
 inserting \Bootstrapping{} just after \cref{line:bbs:cmux}, which is much like \cref{alg:offline}.
When the size $|Q|$ of the states of the DFA $M$ is larger than $2^{N-1}$, we cannot store the index $j$ of the state using one \TRLWE ciphertext (\cref{line:bbs:trivial}).
We can relax this restriction by using multiple \TRLWE ciphertexts for  $\mathbf{c}_{0,q}^{T_i}$  and  $\mathbf{c}_{i+1}^{\mathrm{cur}}$.

\begin{table}[tb]\centering
    \caption{Complexity of the proposed algorithms with respect to the number $|Q|$ of the states of the DFA and the size $|\phi|$ of the LTL formula. For \BlockStream{}, we show the complexity \emph{before} the relaxation.}\label{tbl:complexity-analysis}\scriptsize
    \begin{tabular}{cc|ccc|c}\toprule
        \multirow{2}{*}{Algorithm} & \multirow{2}{*}{w.r.t.} & \multicolumn{3}{c|}{Number of Applications} & \multirow{2}{*}{Space}\\
        & & \CMux{} & \Bootstrapping{} & \CircuitBootstrapping{} & \\\midrule
        \multirow{2}{*}{\FHOffline} & DFA & $O(n|Q|)$ & $O(n|Q|/\BootstrappingInterval)$ & --- & $O(|Q|)$\\
        & LTL & $O(n2^{2^{|\phi|}})$ & $O(n2^{2^{|\phi|}}/\BootstrappingInterval)$ & --- & $O(2^{2^{|\phi|}})$\\\midrule
        \multirow{2}{*}{\ReverseStream} & DFA & $O(n2^{|Q|})$ & $O(n2^{|Q|}/\BootstrappingInterval)$ & --- & $O(2^{|Q|})$\\
        & LTL & $O(n2^{|\phi|})$ & $O(n2^{|\phi|}/\BootstrappingInterval)$ & --- & $O(2^{|\phi|})$\\\midrule
        \multirow{2}{*}{\BlockStream} & DFA & $O(n|Q|)$ & --- & $O((n\log |Q|)/\BlockSize)$ & $O(|Q|)$\\
        & LTL & $O(n2^{2^{|\phi|}})$ & --- & $O(n2^{|\phi|}/\BlockSize)$ & $O(2^{2^{|\phi|}})$\\\bottomrule
    \end{tabular}
\end{table}

\subsection{Complexity Analysis}\label{subsec:complexity_analysis}

\cref{tbl:complexity-analysis} summarizes the complexity of our algorithms with respect to both the number $|Q|$ of the states of the DFA and the size $|\phi|$ of the LTL formula.
We note that, for \BlockStream, we do not relax the above assumptions
for simplicity.
Notice that the number of applications of the homomorphic operations is linear to the length $n$ of the monitored ciphertext. Moreover, the space complexity is independent of $n$.
This shows that our algorithms satisfy the properties essential to good online monitoring;
\begin{oneenumeration}
 \item they only store the minimum of data, and
 \item they run quickly enough under a real-time setting
\end{oneenumeration}~\cite{DBLP:series/lncs/BartocciDDFMNS18}.

The time and the space complexity of \FHOffline and \BlockStream are linear to $|Q|$.
Moreover, in these algorithms, when the $i$-th monitored ciphertext is consumed, only the states reachable by a word of length $i$ are considered, which often makes the scalability even better.
In contrast, the time and the space complexity of \ReverseStream is exponential to $|Q|$. 
This is because of the worst-case size of the reversed DFA due to the powerset construction.
Since the size of the reversed DFA is usually reasonably small,
the practical scalability of \ReverseStream is also much better, which is observed through the experiments in~\cref{sec:experiment}.

For \FHOffline and \BlockStream,
$|Q|$ is \emph{doubly} exponential to $|\phi|$ because we first
convert $\phi$ to an NFA (one exponential) and then construct a DFA from the NFA (second exponential).
In contrast, for \ReverseStream, it is known that we can construct a reversed DFA for $\phi$ of the size of at most \emph{singly} exponential to $|\phi|$~\cite{ijcai2020-690}.
Note that, in a practical scenario exemplified in~\cref{sec:experiment},
the size of the DFA constructed from $\phi$ is expected to be much smaller than the worst one.

\section{Oblivious Online LTL Monitoring}\label{sec:oblivious-ltl-mon}

In this section, we formalize the scheme of oblivious online LTL monitoring.
We consider a two-party setting with a client and a server and refer to 
the client and the server as Alice and Bob, respectively.
Here, we assume that Alice has private data sequence 
$w = \sigma_1\sigma_2\dots\sigma_n$ to be monitored
where $\sigma_i\in 2^{\AP}$ for each $i \ge 1$.
Meanwhile, Bob has a private LTL formula $\phi$.
The purpose of oblivious online LTL monitoring is to let Alice know
if $\sigma_1\sigma_2\dots\sigma_{i}\models\phi$
for each $i \ge 1$,
while keeping the privacy of Alice and Bob.

\subsection{Threat Model}
We assume that Alice is \emph{malicious},
i.e., Alice can deviate arbitrarily from the protocol
to try to learn $\phi$.
We also assume that Bob is \emph{honest-but-curious},
i.e., Bob correctly follows the protocol,
but he tries to learn $w$ from the information he obtains from the protocol execution.
We do not assume that Bob is malicious in the present paper; a protocol that is secure against malicious Bob requires more sophisticated primitives such as zero-knowledge proofs and is left as future work.

\paragraph{Public and Private Data.}
We assume that the TFHE parameters,
the parameters of our\begin{ExtendedVersion} proposed\end{ExtendedVersion} 
algorithms (e.g., $\BootstrappingInterval$ and $B$),
Alice's public key $\PK$, and Alice's bootstrapping key $\BK$
are public to both parties.
The input $w$ and the monitoring result are private for Alice,
and the LTL formula $\phi$ is private for Bob.

\subsection{Protocol Flow}\label{subsec:protocol-flow}

The protocol flow of oblivious online LTL monitoring is shown in \cref{fig:protocol}.
It takes $\sigma_1,\sigma_2,\dots,\sigma_n$, $\phi$, and $b\in\B$ as its parameters,
where $b$ is a flag that indicates the algorithm Bob uses: \ReverseStream ($b=0$) or \BlockStream ($b=1$).
After generating her secret key and sending the corresponding public and bootstrapping key to Bob (\crefrange{line:protocol:gen-sk}{line:protocol:send-pk-and-bk}),
Alice encrypts her inputs into ciphertexts and sends the ciphertexts to Bob one by one (\crefrange{line:protocol:for0}{line:protocol:send-ctxt}).
In contrast, Bob first converts his LTL formula $\phi$ to a binary DFA $M$ (\cref{line:protocol:convert-phi}).
Then, Bob serially feeds the received ciphertexts from Alice to
\ReverseStream or \BlockStream (\cref{line:protocol:feed-ctxt}) and
returns the encrypted output of the algorithm to Alice (\crefrange{line:protocol:output0}{line:protocol:output1}).

Note that, although the alphabet of a DFA constructed from an LTL formula is $2^{\AP}$~\cite{DBLP:journals/fmsd/TabakovRV12},
our proposed algorithms require a binary DFA.\@
Thus, in \cref{line:protocol:convert-phi},
we convert the DFA constructed from $\phi$ to a binary DFA $M$
by inserting auxiliary states.
Besides, in \cref{line:protocol:encode},
we encode an observation $\sigma_i \in 2^{\AP}$ by a sequence $\sigma'_i\coloneqq({\sigma'}_i^1, {\sigma'}_i^2, \dots, {\sigma'}_i^{|AP|}) \in \B^{|AP|}$
such that $p_j \in \sigma_i$ if and only if ${\sigma'}_i^j$ is true,
where $\AP = \{p_1, \dots, p_{|\AP|}\}$. 
We also note that, taking this encoding into account,
we need to properly set the parameters for \BlockStream{} to generate an output
for each $|\AP|$-size block of Alice's inputs,
i.e., $B$ is taken to be equal to $|\AP|$.

Here, we provide brief sketches of the correctness and security analysis of the proposed protocol.
See \crefAppendix{correctness-and-security-of-protocol} for detailed explanations and proofs.

\begin{figure}[tb]
\setlength{\interspacetitleruled}{0pt}%
\setlength{\algotitleheightrule}{0pt}%
\begin{algorithm}[H]\scriptsize
\Input{%
    Alice's private inputs $\sigma_1, \sigma_2, \dots, \sigma_n\in 2^\AP$,
    Bob's private LTL formula $\phi$,
    and $b\in\B$}
\Output{%
    For every $i\in\{1,2,\dots n\}$,
    Alice's private output representing
    $\sigma_1\sigma_2\dots\sigma_{i}\models\phi$}

Alice generates her secret key $\SK$.\;\label{line:protocol:gen-sk}
Alice generates her public key $\PK$ and bootstrapping key $\BK$ from $\SK$.\;\label{line:protocol:gen-pk-and-bk}
Alice sends $\PK$ and $\BK$ to Bob.\;\label{line:protocol:send-pk-and-bk}
Bob converts $\phi$ to a binary DFA $M=(Q, \Sigma=\B, \delta, q_0, F)$.\;\label{line:protocol:convert-phi}
\For{$i=1,2,\dots,n$}{\label{line:protocol:for0}%
    Alice encodes $\sigma_i$ to a sequence $\sigma'_i \coloneqq ({\sigma'}_i^1,{\sigma'}_i^{2},\dots,{\sigma'}_i^{|\AP|})\in\B^{|\AP|}$.\;\label{line:protocol:encode}
    Alice calculates $d_i \coloneqq (\Enc({\sigma'}_i^1),\Enc({\sigma'}_i^2),\dots\Enc({\sigma'}_i^{|\AP|}))$.\;\label{line:protocol:encrypt}
    Alice sends $d_i$ to Bob.\;\label{line:protocol:send-ctxt}
    Bob feeds the elements of $d_i$ to \ReverseStream (if $b=0$) or \BlockStream (if $b=1$).\;\label{line:protocol:feed-ctxt}
    \tcp{$\sigma'_1\cdot\sigma'_2\cdots\sigma'_i$ refers ${\sigma'}_1^{1}\dots{\sigma'}_1^{|\AP|}{\sigma'}_2^{1}\dots{\sigma'}_2^{|\AP|}{\sigma'}_3^{1}\dots{\sigma'}_i^{|\AP|}$.}
    Bob obtains the output \TLWE ciphertext $c$ produced by the algorithm,
    where $\Dec(c)=M(\sigma'_1\cdot\sigma'_2\cdot\cdots\cdot\sigma'_i)$.\;\label{line:protocol:obtain-output}
    \label{line:protocol:output0}%
    Bob randomizes $c$ to obtain $c'$ so that $\Dec(c) = \Dec(c')$.\;\label{line:protocol:randomize}
    Bob sends $c'$ to Alice.\;
    Alice calculates $\Dec(c')$ to obtain the result in plaintext.\;\label{line:protocol:output1}
}
\end{algorithm}
\caption{Protocol of oblivious online LTL monitoring.}\label{fig:protocol}
\end{figure}

\noindent{\textbf{Correctness.}}
We can show that Alice obtains correct results in our protocol
directly by \cref{thm:correctness-reversed} and \cref{thm:correctness-bbs}.

\noindent{\textbf{Security.}} Intuitively, after the execution of the protocol described in \cref{fig:protocol},
Alice should learn $M(\sigma'_1\cdot\sigma'_2\cdots\sigma'_{i})$
for every $i\in\{1,2,\dots,n\}$ but nothing else.
Besides, Bob should learn the input size $n$ but nothing else.

\noindent{\textit{Privacy for Alice.}}
We observe that Bob only obtains $\Enc({\sigma'}_i^j)$ from Alice for each $i\in\{1,2,\dots,n\}$ and $j\in\{1,2,\dots,|\AP|\}$.
Therefore, we need to show that Bob learns nothing from the ciphertexts generated by Alice. 
Since TFHE provides IND-CPA security~\cite{bellare2005introduction},
we can easily guarantee the client's privacy for Alice.

\noindent{\textit{Privacy for Bob.}}
The privacy guarantee for Bob is more complex than that for Alice. 
Here, Alice obtains $\sigma'_1,\sigma'_2,\dots,\sigma'_n$ and the results $M(\sigma'_1\cdot\sigma'_2\cdots\sigma'_{i})$ for every $i\in\{1,2,\dots,n\}$ in plaintext.
In the protocol (\cref{fig:protocol}), Alice does not obtain $\phi,M$ themselves or their sizes,
and it is known that a finite number of checking $M(w)$
cannot uniquely identify $M$
if any additional information (e.g., $|M|$) is not given~\cite{Moore1956,DBLP:journals/iandc/Angluin81}.
Thus, it is impossible for Alice to identify $M$ (or $\phi$) from the input/output pairs.

Nonetheless, to fully guarantee the model privacy of Bob,
we also need to show that, when Alice inspects the result ciphertext $c'$,
it is impossible for Alice to know Bob's specification, i.e.,
what homomorphic operations were applied by Bob to obtain $c'$.
A TLWE ciphertext contains a random nonce and a noise term.
By randomizing $c$ properly in \cref{line:protocol:randomize},
we ensure that the random nonce of $c'$ is not biased~\cite{DBLP:journals/jacm/Regev09}.
By assuming SRL security~\cite{DBLP:journals/iacr/BrakerskiDGM20a,DBLP:journals/iacr/GayP20} over TFHE, we can ensure that there is no information leakage regarding Bob's specifications through the noise bias.
A more detailed discussion is in \crefAppendix{correctness-and-security-of-protocol}.

\section{Experiments}\label{sec:experiment}

We experimentally evaluated the proposed algorithms (\ReverseStream{} and \BlockStream{}) and protocol.
We pose the following two research questions:
\begin{description}
    \item[RQ1] Are the proposed algorithms scalable with respect to the size of the monitored ciphertexts and that of the DFA?
    \item[RQ2] Are the proposed algorithms fast enough in a realistic monitoring scenario?
    \item[RQ3] Does a standard IoT device have sufficient computational power acting as a client in the proposed protocol?
\end{description}
To answer RQ1, we conducted an experiment with our original benchmark where the length of the monitored ciphertexts and the size of the DFA are configurable (\cref{subsec:experiment-simple}).
To answer RQ2 and RQ3, we conducted a case study on blood glucose monitoring; we monitored blood glucose data obtained by simglucose~\footnote{\url{https://github.com/jxx123/simglucose}} against specifications taken from~\cite{DBLP:conf/iotdi/YoungCGPF18,DBLP:conf/rv/CameronFMS15} (\cref{subsec:experiment-monitoring-bg}).
To answer RQ3,
we measured the time spent on the encryption of plaintexts, %
which is the heaviest task for a client during the execution of the online protocol.

We implemented our algorithms in C++20.
Our implementation is publicly available\footnote{Our implementation is uploaded to \SrcURL.\label{page:src-url}}.
We used Spot~\cite{DBLP:conf/atva/Duret-LutzLFMRX16} to convert a safety LTL formula to a DFA.\@
We also used a Spot's utility program \verb|ltlfilt| to calculate the size of an LTL formula\footnote{%
We desugared a formula by \verb|ltlfilt| with option \verb|--unabbreviate="eFGiMRW^"| and counted the number of the characters.}.
We used TFHEpp~\cite{DBLP:conf/uss/MatsuokaBMS021} as the TFHE library.
We used $N=1024$ as the size of the message represented by one \TRLWE{} ciphertext,
which is a parameter of TFHE.\@
The complete TFHE parameters we used are shown in \crefAppendix{TFHE_parameters}.

For RQ1 and RQ2, we ran experiments on a workstation with Intel Xeon Silver 4216 (3.2GHz; 32 cores and 64 threads in total), 128GiB RAM, and Ubuntu 20.04.2 LTS.
We ran each instance of the experiment setting five times and reported the average.
We measured the time to consume all of the monitored ciphertexts in the main loop of each algorithm,
i.e., in \crefrange{line:reversed:while0}{line:reversed:while1} in \ReverseStream
and in \crefrange{line:bbs:while0}{line:bbs:while1} in \BlockStream.

For RQ3, we ran experiments on two single-board computers with and without Advanced Encryption Standard (AES)~\cite{daemen1999aes} hardware accelerator.
ROCK64\begin{ExtendedVersion}~\footnote{\url{https://www.pine64.org/devices/single-board-computers/rock64/}}\end{ExtendedVersion} has ARM Cortex A53 CPU cores (1.5GHz; 4 cores) with AES hardware accelerator and 4GiB RAM.\@
Raspberry Pi 4\begin{ExtendedVersion}~\footnote{\url{https://www.raspberrypi.com/products/raspberry-pi-4-model-b/}}\end{ExtendedVersion} has ARM Cortex A72 CPU cores (1.5GHz; 4 cores) without AES hardware accelerator and 4GiB RAM.\@

\subsection{RQ1: Scalability}\label{subsec:experiment-simple}

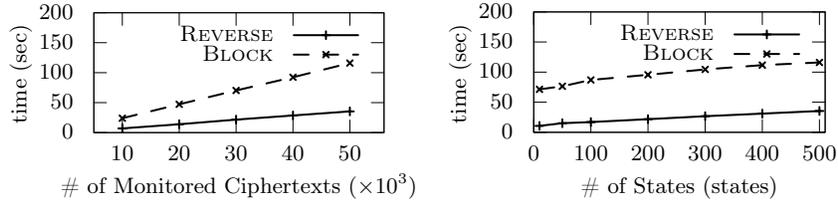
\begin{figure}[tb]
    \centering
    \resizebox{.47\textwidth}{!}{\begin{tikzpicture}[gnuplot]
\path (0.000,0.000) rectangle (6.000,3.000);
\gpcolor{color=gp lt color border}
\gpsetlinetype{gp lt border}
\gpsetdashtype{gp dt solid}
\gpsetlinewidth{1.00}
\draw[gp path] (1.320,0.985)--(1.500,0.985);
\draw[gp path] (5.447,0.985)--(5.267,0.985);
\node[gp node right] at (1.136,0.985) {$0$};
\draw[gp path] (1.320,1.412)--(1.500,1.412);
\draw[gp path] (5.447,1.412)--(5.267,1.412);
\node[gp node right] at (1.136,1.412) {$50$};
\draw[gp path] (1.320,1.838)--(1.500,1.838);
\draw[gp path] (5.447,1.838)--(5.267,1.838);
\node[gp node right] at (1.136,1.838) {$100$};
\draw[gp path] (1.320,2.265)--(1.500,2.265);
\draw[gp path] (5.447,2.265)--(5.267,2.265);
\node[gp node right] at (1.136,2.265) {$150$};
\draw[gp path] (1.320,2.691)--(1.500,2.691);
\draw[gp path] (5.447,2.691)--(5.267,2.691);
\node[gp node right] at (1.136,2.691) {$200$};
\draw[gp path] (1.725,0.985)--(1.725,1.165);
\draw[gp path] (1.725,2.691)--(1.725,2.511);
\node[gp node center] at (1.725,0.677) {$10$};
\draw[gp path] (2.534,0.985)--(2.534,1.165);
\draw[gp path] (2.534,2.691)--(2.534,2.511);
\node[gp node center] at (2.534,0.677) {$20$};
\draw[gp path] (3.343,0.985)--(3.343,1.165);
\draw[gp path] (3.343,2.691)--(3.343,2.511);
\node[gp node center] at (3.343,0.677) {$30$};
\draw[gp path] (4.152,0.985)--(4.152,1.165);
\draw[gp path] (4.152,2.691)--(4.152,2.511);
\node[gp node center] at (4.152,0.677) {$40$};
\draw[gp path] (4.961,0.985)--(4.961,1.165);
\draw[gp path] (4.961,2.691)--(4.961,2.511);
\node[gp node center] at (4.961,0.677) {$50$};
\draw[gp path] (1.320,2.691)--(1.320,0.985)--(5.447,0.985)--(5.447,2.691)--cycle;
\node[gp node center,rotate=-270] at (0.292,1.838) {time (sec)};
\node[gp node center] at (3.383,0.215) {\# of Monitored Ciphertexts ($\times 10^3$)};
\node[gp node right] at (3.979,2.372) {\ReverseStream};
\gpsetlinewidth{2.00}
\draw[gp path] (4.163,2.372)--(5.079,2.372);
\draw[gp path] (1.725,1.045)--(2.534,1.104)--(3.343,1.170)--(4.152,1.229)--(4.961,1.287);
\gpsetpointsize{4.00}
\gppoint{gp mark 1}{(1.725,1.045)}
\gppoint{gp mark 1}{(2.534,1.104)}
\gppoint{gp mark 1}{(3.343,1.170)}
\gppoint{gp mark 1}{(4.152,1.229)}
\gppoint{gp mark 1}{(4.961,1.287)}
\gppoint{gp mark 1}{(4.621,2.372)}
\node[gp node right] at (3.979,2.095) {\BlockStream};
\gpsetdashtype{gp dt 2}
\draw[gp path] (4.163,2.095)--(5.079,2.095);
\draw[gp path] (1.725,1.190)--(2.534,1.388)--(3.343,1.585)--(4.152,1.773)--(4.961,1.975);
\gppoint{gp mark 2}{(1.725,1.190)}
\gppoint{gp mark 2}{(2.534,1.388)}
\gppoint{gp mark 2}{(3.343,1.585)}
\gppoint{gp mark 2}{(4.152,1.773)}
\gppoint{gp mark 2}{(4.961,1.975)}
\gppoint{gp mark 2}{(4.621,2.095)}
\gpsetdashtype{gp dt solid}
\gpsetlinewidth{1.00}
\draw[gp path] (1.320,2.691)--(1.320,0.985)--(5.447,0.985)--(5.447,2.691)--cycle;
\gpdefrectangularnode{gp plot 1}{\pgfpoint{1.320cm}{0.985cm}}{\pgfpoint{5.447cm}{2.691cm}}
\end{tikzpicture}}
    \resizebox{0.47\textwidth}{!}{\begin{tikzpicture}[gnuplot]
\path (0.000,0.000) rectangle (6.000,3.000);
\gpcolor{color=gp lt color border}
\gpsetlinetype{gp lt border}
\gpsetdashtype{gp dt solid}
\gpsetlinewidth{1.00}
\draw[gp path] (1.320,0.985)--(1.500,0.985);
\draw[gp path] (5.447,0.985)--(5.267,0.985);
\node[gp node right] at (1.136,0.985) {$0$};
\draw[gp path] (1.320,1.412)--(1.500,1.412);
\draw[gp path] (5.447,1.412)--(5.267,1.412);
\node[gp node right] at (1.136,1.412) {$50$};
\draw[gp path] (1.320,1.838)--(1.500,1.838);
\draw[gp path] (5.447,1.838)--(5.267,1.838);
\node[gp node right] at (1.136,1.838) {$100$};
\draw[gp path] (1.320,2.265)--(1.500,2.265);
\draw[gp path] (5.447,2.265)--(5.267,2.265);
\node[gp node right] at (1.136,2.265) {$150$};
\draw[gp path] (1.320,2.691)--(1.500,2.691);
\draw[gp path] (5.447,2.691)--(5.267,2.691);
\node[gp node right] at (1.136,2.691) {$200$};
\draw[gp path] (1.320,0.985)--(1.320,1.165);
\draw[gp path] (1.320,2.691)--(1.320,2.511);
\node[gp node center] at (1.320,0.677) {$0$};
\draw[gp path] (2.129,0.985)--(2.129,1.165);
\draw[gp path] (2.129,2.691)--(2.129,2.511);
\node[gp node center] at (2.129,0.677) {$100$};
\draw[gp path] (2.938,0.985)--(2.938,1.165);
\draw[gp path] (2.938,2.691)--(2.938,2.511);
\node[gp node center] at (2.938,0.677) {$200$};
\draw[gp path] (3.748,0.985)--(3.748,1.165);
\draw[gp path] (3.748,2.691)--(3.748,2.511);
\node[gp node center] at (3.748,0.677) {$300$};
\draw[gp path] (4.557,0.985)--(4.557,1.165);
\draw[gp path] (4.557,2.691)--(4.557,2.511);
\node[gp node center] at (4.557,0.677) {$400$};
\draw[gp path] (5.366,0.985)--(5.366,1.165);
\draw[gp path] (5.366,2.691)--(5.366,2.511);
\node[gp node center] at (5.366,0.677) {$500$};
\draw[gp path] (1.320,2.691)--(1.320,0.985)--(5.447,0.985)--(5.447,2.691)--cycle;
\node[gp node center,rotate=-270] at (0.292,1.838) {time (sec)};
\node[gp node center] at (3.383,0.215) {\# of States (states)};
\node[gp node right] at (3.979,2.372) {\ReverseStream};
\gpsetlinewidth{2.00}
\draw[gp path] (4.163,2.372)--(5.079,2.372);
\draw[gp path] (1.401,1.076)--(1.725,1.114)--(2.129,1.130)--(2.938,1.171)--(3.748,1.213)%
  --(4.557,1.250)--(5.366,1.287);
\gpsetpointsize{4.00}
\gppoint{gp mark 1}{(1.401,1.076)}
\gppoint{gp mark 1}{(1.725,1.114)}
\gppoint{gp mark 1}{(2.129,1.130)}
\gppoint{gp mark 1}{(2.938,1.171)}
\gppoint{gp mark 1}{(3.748,1.213)}
\gppoint{gp mark 1}{(4.557,1.250)}
\gppoint{gp mark 1}{(5.366,1.287)}
\gppoint{gp mark 1}{(4.621,2.372)}
\node[gp node right] at (3.979,2.095) {\BlockStream};
\gpsetdashtype{gp dt 2}
\draw[gp path] (4.163,2.095)--(5.079,2.095);
\draw[gp path] (1.401,1.594)--(1.725,1.637)--(2.129,1.727)--(2.938,1.800)--(3.748,1.876)%
  --(4.557,1.937)--(5.366,1.975);
\gppoint{gp mark 2}{(1.401,1.594)}
\gppoint{gp mark 2}{(1.725,1.637)}
\gppoint{gp mark 2}{(2.129,1.727)}
\gppoint{gp mark 2}{(2.938,1.800)}
\gppoint{gp mark 2}{(3.748,1.876)}
\gppoint{gp mark 2}{(4.557,1.937)}
\gppoint{gp mark 2}{(5.366,1.975)}
\gppoint{gp mark 2}{(4.621,2.095)}
\gpsetdashtype{gp dt solid}
\gpsetlinewidth{1.00}
\draw[gp path] (1.320,2.691)--(1.320,0.985)--(5.447,0.985)--(5.447,2.691)--cycle;
\gpdefrectangularnode{gp plot 1}{\pgfpoint{1.320cm}{0.985cm}}{\pgfpoint{5.447cm}{2.691cm}}
\end{tikzpicture}}
    \caption{Experimental results of $M_m$. The left figure shows runtimes when the number of states (i.e., $m$) is fixed to $500$, while the right one is when the number of monitored ciphertexts (i.e., $n$) is fixed to $50000$.}
    \label{fig:artificial-dfa}
\end{figure}

\noindent{\textbf{Experimental Setup.}}
In the experiments to answer RQ1,
we used a simple binary DFA $M_m$,
which accepts a word $w$ if and only if the number of the appearance of $1$ in $w$ is a multiple of $m$.
The number of the states of $M_m$ is $m$.

Our experiments are twofold.
In the first experiment, we fixed the DFA size $m$ to $500$
and increased the size $n$ of the input word $w$ from $10000$ to $50000$.
In the second experiment, we fixed $n=50000$ and changed $m$ from $10$ to $500$.
The parameters we used are $\BootstrappingInterval = 30000$ and $B = 150$.

\noindent{\textbf{Results and Discussion.}}
\cref{fig:artificial-dfa} shows the results of the experiments.
In the left plot of \cref{fig:artificial-dfa}, we observe that the runtimes of both algorithms are linear to the length of the monitored ciphertexts.
This coincides with the complexity analysis in \cref{subsec:complexity_analysis}.

In the right plot of \cref{fig:artificial-dfa}, we observe that the runtimes of both algorithms are at most linear to the number of the states.
For \BlockStream, this coincides with the complexity analysis in \cref{subsec:complexity_analysis}.
In contrast, this is much more efficient than the exponential complexity of \ReverseStream with respect to $|Q|$.
This is because the size of the reversed DFA does not increase.

In both plots of \cref{fig:artificial-dfa}, we observe that \ReverseStream is faster than \BlockStream.
Moreover, in the left plot of \cref{fig:artificial-dfa}, the curve of \BlockStream is steeper than that of \ReverseStream.
This is because
\begin{oneenumeration}
\item the reversed DFA $\Rev{M_m}$ has the same size as $M_m$, 
\item \CircuitBootstrapping{} is about ten times slower than \Bootstrapping, and
\item $\BootstrappingInterval$ is much larger than $\BlockSize$.
\end{oneenumeration}

Overall, our experiment results confirm the complexity analysis in \cref{subsec:complexity_analysis}.
Moreover, the practical scalability of \ReverseStream with respect to the DFA size is much better than the worst case, at least for this benchmark.
Therefore, we answer RQ1 affirmatively.

\subsection{RQ2 and RQ3: Case Study on Blood Glucose Monitoring}\label{subsec:experiment-monitoring-bg}

\noindent{\textbf{Experimental Setup.}}
To answer RQ2,
we applied \ReverseStream and \BlockStream to the monitoring of blood glucose levels.
The monitored values are generated by simulation of type 1 diabetes patients.
We used the LTL formulae in \cref{tbl:glucose-ltl-formulae}.
These formulae are originally presented as signal temporal logic~\cite{DBLP:conf/formats/MalerN04} formulae~\cite{DBLP:conf/iotdi/YoungCGPF18,DBLP:conf/rv/CameronFMS15}, and
we obtained the LTL formulae in \cref{tbl:glucose-ltl-formulae} by discrete sampling.

To simulate blood glucose levels of type 1 diabetes patients,
we adopted simglucose, which is a Python implementation of UVA/Padova Type 1 Diabetes Simulator~\cite{man2014uva}.
We recorded the blood glucose levels every one minute\footnote{Current continuous glucose monitors (e.g., Dexcom G4 PLATINUM) record blood glucose levels every few minutes, and our sampling interval is realistic.}
and encoded each of them in nine bits.
For $\psi_1,\psi_2,\psi_4$, we used 720 minutes of the simulated values. %
For $\phi_1,\phi_4,\phi_5$, we used seven days of the values.
The parameters we used are $\BootstrappingInterval = 30000$, $B = 9$.

To answer RQ3, we encrypted plaintexts into TRGSW ciphertexts 1000 times using two single-board computers (ROCK64 and Raspberry Pi 4) and reported the average runtime.

\begin{table}[tb]\centering\scriptsize
        \caption{The safety LTL formulae used in our experiments. $\psi_1,\psi_2,\psi_4$ are originally from~\cite{DBLP:conf/rv/CameronFMS15}, and $\phi_1$, $\phi_4$, and $\phi_5$ are originally from~\cite{DBLP:conf/iotdi/YoungCGPF18}.}\label{tbl:glucose-ltl-formulae}
        \begin{tabularx}{\linewidth}{>{\centering\arraybackslash}p{5em}|X}\toprule
	     &\multicolumn{1}{c}{LTL formula} \\\midrule

                $\psi_1$ &
                $
                    \G_{[100, 700]}(p_8\lor p_9\lor (p_4\land p_7)\lor (p_5\land p_7)\lor (p_6\land p_7)\lor (p_2\land p_3\land p_7))
                $\\

                $\psi_2$ &
                $
                    \G_{[100, 700]}(\lnot p_9\lor (\lnot p_7\land \lnot p_8)\lor (\lnot p_5\land \lnot p_6\land \lnot p_8)\lor (\lnot p_4\land \lnot p_6\land \lnot p_8)\lor (\lnot p_3\land \lnot p_6\land \lnot p_8)\lor (\lnot p_2\land \lnot p_6\land \lnot p_8)\lor (\lnot p_1\land \lnot p_6\land \lnot p_8))
                $\\

                $\psi_4$ &
                $
                    \G_{[600, 700]}((\lnot p_8\land \lnot p_9)\lor (\lnot p_7\land \lnot p_9)\lor (\lnot p_4\land \lnot p_5\land \lnot p_6\land \lnot p_9)\lor (\lnot p_1\land \lnot p_2\land \lnot p_3\land \lnot p_5\land \lnot p_6\land \lnot p_9))
                $\\\midrule

                $\phi_1$ &
                $
                    \G((\lnot p_6\land \lnot p_7\land p_8\land \lnot p_9)\lor (\lnot p_5\land \lnot p_7\land p_8\land \lnot p_9)\lor (\lnot p_3\land \lnot p_4\land \lnot p_7\land p_8\land \lnot p_9)\lor (p_4\land p_7\land \lnot p_8\land \lnot p_9)\lor (p_5\land p_7\land \lnot p_8\land \lnot p_9)\lor (p_6\land p_7\land \lnot p_8\land \lnot p_9)\lor (p_1\land p_2\land p_3\land p_7\land \lnot p_8\land \lnot p_9))
                $\\

                $\phi_4$ &
                $
                    \G((\lnot p_7\land \lnot p_8\land \lnot p_9) \implies \F_{[0, 25]}(p_7\lor p_8\lor p_9))
                $ \\

                $\phi_5$ &
                $
                    \G(p_9\lor (p_3\land p_7\land p_8)\lor (p_4\land p_7\land p_8)\lor (p_5\land p_7\land p_8)\lor (p_6\land p_7\land p_8) \implies \F_{[0, 25]}((\lnot p_8\land \lnot p_9)\lor (\lnot p_7\land \lnot p_9)\lor (\lnot p_3\land \lnot p_4\land \lnot p_5\land \lnot p_6\land \lnot p_9)))
                $\\\bottomrule

        \end{tabularx}
\centering
\caption{Experimental results of blood glucose monitoring, where $Q$ is the state space of the monitoring DFA and $\Rev{Q}$ is the state space of the reversed DFA.}%
\label{tbl:glucose-experiment-result}
\scriptsize
\begin{tabular}{cccccc|r|r}\toprule
Formula $\phi$ & $|\phi|$ & $|Q|$
& $|\Rev{Q}|$ &
\begin{tabular}{@{}c@{}}\# of blood glucose values\end{tabular} & Algorithm & \begin{tabular}{@{}c@{}}Runtime (s)\end{tabular} & \begin{tabular}{@{}c@{}}Mean Runtime (ms/value)\end{tabular}   \\\midrule
\multirow{2}{*}{$\psi_1$} & \multirow{2}{*}{40963}   & \multirow{2}{*}{10524} & \multirow{2}{*}{2712974}    & \multirow{2}{*}{721}   & \ReverseStream    & $16021.06$ & $22220.62$ \\
                          &                          &                        &                             &                        & \BlockStream      & $132.68$ & $184.02$ \\\midrule
\multirow{2}{*}{$\psi_2$} & \multirow{2}{*}{75220}   & \multirow{2}{*}{11126} & \multirow{2}{*}{2885376}    & \multirow{2}{*}{721}   & \ReverseStream    & $17035.05$ & $23626.97$ \\
                          &                          &                        &                             &                        & \BlockStream      & $131.53$ & $182.43$ \\\midrule
\multirow{2}{*}{$\psi_4$} & \multirow{2}{*}{10392}   & \multirow{2}{*}{7026}  & \multirow{2}{*}{---}        & \multirow{2}{*}{721}   & \ReverseStream    & ---               & ---                 \\
                          &                          &                        &                             &                        & \BlockStream      & $35.42$ & $49.12$ \\\midrule\midrule

\multirow{2}{*}{$\phi_1$} & \multirow{2}{*}{195}     & \multirow{2}{*}{21}    & \multirow{2}{*}{20}         & \multirow{2}{*}{10081} & \ReverseStream    & $22.33$ & $2.21$ \\
                          &                          &                        &                             &                        & \BlockStream      & $1741.15$ & $172.72$ \\\midrule
\multirow{2}{*}{$\phi_4$} & \multirow{2}{*}{494}     & \multirow{2}{*}{237}   & \multirow{2}{*}{237}        & \multirow{2}{*}{10081} & \ReverseStream    & $42.23$ & $4.19$ \\
                          &                          &                        &                             &                        & \BlockStream      & $2073.45$ & $205.68$ \\\midrule
\multirow{2}{*}{$\phi_5$} & \multirow{2}{*}{1719}    & \multirow{2}{*}{390}   & \multirow{2}{*}{390}        & \multirow{2}{*}{10081} & \ReverseStream    & $54.87$ & $5.44$ \\
                          &                          &                        &                             &                        & \BlockStream      & $2084.50$ & $206.78$ \\\bottomrule
\end{tabular}
\end{table}

\noindent{\textbf{Results and Discussion (RQ2).}}
The results of the experiments are shown in \cref{tbl:glucose-experiment-result}.
The result for $\psi_4$ with \ReverseStream is missing because
the reversed DFA for $\psi_4$ is too huge, and its construction was aborted due to the memory limit.

Although the size of the reversed DFA was large for $\psi_1$ and $\psi_2$,
in all the cases, we observe that both \ReverseStream and \BlockStream
took at most 24 seconds to process each blood glucose value on average.
This is partly because $|Q|$ and $|Q^R|$ are not so large in comparison with the upper bound described in \cref{subsec:complexity_analysis},
i.e., doubly or singly exponential to $|\phi|$, respectively.
Since each value is recorded every one minute, at least on average, both algorithms finished processing each value
before the next measured value arrived, i.e., any congestion did not occur.
Therefore, our experiment results confirm that,
in a practical scenario of blood glucose monitoring,
both of our proposed algorithms are fast enough to be used in the online setting,
and we answer RQ2 affirmatively.

We also observe that average runtimes of $\psi_1,\psi_2,\psi_4$ and $\phi_1,\phi_4,\phi_5$
with \BlockStream are comparable, although the monitoring DFA of $\psi_1,\psi_2,\psi_4$
are significantly larger than those of $\phi_1,\phi_4,\phi_5$.
This is because the numbers of the reachable states
during execution are similar among these cases (from $1$ up to $27$ states).
As we mentioned in \cref{subsec:complexity_analysis},
\BlockStream only considers the states reachable
by a word of length $i$ when the $i$-th monitored ciphertext is consumed, and thus,
it ran much faster even if the monitoring DFA is large.

\noindent{\textbf{Results and Discussion (RQ3).}}
It took 40.41 and 1470.33 ms on average to encrypt a value of blood glucose (i.e., nine bits) on ROCK64 and Raspberry Pi 4, respectively.
Since each value is sampled every one minute,
our experiment results confirm that both machines are fast enough to be used in an
online setting. Therefore, we answer RQ3 affirmatively.

We also observe that encryption on ROCK64 is more than 35 times faster than that on Raspberry Pi 4.
This is mainly because of the hardware accelerator for AES, which is used in TFHEpp to generate TRGSW
ciphertexts.

\section{Conclusion}\label{sec:conclusion}

We presented the first oblivious online LTL monitoring protocol up to our knowledge.
Our protocol allows online LTL monitoring concealing
\begin{oneenumeration}
\item the client's monitored inputs from the server and
\item the server's LTL specification from the client.
\end{oneenumeration}
We proposed two online algorithms (\ReverseStream and \BlockStream)
using an FHE scheme called TFHE.
In addition to the complexity analysis, we 
experimentally confirmed the scalability and practicality of our algorithms
with an artificial benchmark and a case study on blood glucose level monitoring.

Our immediate future work is to extend our approaches to LTL semantics
with multiple values, e.g., LTL${}_3$~\cite{DBLP:journals/tosem/BauerLS11} and
rLTL~\cite{DBLP:conf/hybrid/MascleNSTW020}.
Extension to monitoring continuous-time signals, e.g., against an STL~\cite{DBLP:conf/formats/MalerN04} formula,
is also future work. 
Another future direction is to conduct a more realistic case study of our framework with
actual IoT devices. %

\noindent\textbf{Acknowledgements.}
This work was partially supported by JST ACT-X Grant No.~JPMJAX200U, JSPS KAKENHI Grant No.~22K17873 and 19H04084, and JST CREST Grant No.~JPMJCR19K5, JPMJCR2012, and JPMJCR21M3.

\ifdraft
\pagelimitmarker{19}
\fi
\newpage

\bibliographystyle{splncs04}
\bibliography{ref_dblp_generated,ref_manual}

\begin{ExtendedVersion}
\appendix

\section{Correctness and Security of the Protocol}\label{appendix:correctness-and-security-of-protocol}

\subsection{Correctness}

The correctness of our protocol shown in \cref{fig:protocol} is
formulated as follows:
\begin{theorem}
    Let $\phi$ be Bob's LTL formula and $M$ be the binary DFA converted from $\phi$.
    Let $\sigma_1,\sigma_2,\dots,\sigma_n\in2^{\AP}$ be Alice's inputs and
    $\sigma'_1,\sigma'_2,\dots,\sigma'_n\in\B^{|\AP|}$ be the encoded Alice's inputs.
    We assume
    \begin{oneenumeration}
     \item the protocol uses \ReverseStream (i.e., $b=0$)
           and, for every $i\in\{1,2,\dots,n\}$,
           $c_i$ in \cref{alg:reversed} can be correctly decrypted, or
     \item the protocol uses \BlockStream (i.e., $b=1$)
           and, for every $i\in\{1,2,\dots,n\}$,
           $c_i$ in \cref{alg:bbs} can be correctly decrypted.
    \end{oneenumeration}
    Then, by following the protocol described in \cref{fig:protocol},
    Alice obtains a Boolean value representing if
    $\sigma_1\sigma_2\dots\sigma_{i}\models\phi$
    for every $i\in\{1,2,\dots,n\}$.

\end{theorem}
\begin{proof}[sketch]
    It suffices to show that, for every $i\in\{1,2,\dots, n\}$,
    the decryption of the resulted ciphertext $c'$ in \cref{line:protocol:obtain-output}
    is equal to $M(\sigma'_1\cdot\sigma'_2\cdots\sigma'_i)$, which indicates
    $\sigma_1\sigma_2\dots\sigma_i\models\phi$.
    When executing \cref{line:protocol:obtain-output} in \cref{fig:protocol},
    we can confirm that Bob has already fed
    the monitored ciphertexts $d_1,d_2,\dots,d_i$ to \cref{alg:reversed} or \cref{alg:bbs}.
    Therefore, by \cref{thm:correctness-reversed} and \cref{thm:correctness-bbs},
    the \TLWE ciphertext $c$ produced by these algorithms can be correctly decrypted,
    and we have $\Dec(c)=M(\sigma'_1\cdot\sigma'_2\cdots\sigma'_i)$.
    Moreover, randomization (i.e., adding $\Enc(0)$ to $c$) in \cref{line:protocol:randomize}
    does not change its message, i.e., $\Dec(c') = \Dec(c)$ holds.
    \qed{}
\end{proof}

\subsection{Security}

In this subsection, we formally define the privacy of Alice and Bob, and
based on these definitions, we prove the security of the protocol described in \cref{fig:protocol}.
We refer to \cite{DBLP:conf/tcc/IshaiP07}
for the formal definitions of the privacy of the client (Alice) and the server (Bob).
We note that the privacy of the server requires an additional assumption of TFHE called
\emph{shielded randomness leakage} (SRL) security~\cite{DBLP:journals/iacr/BrakerskiDGM20a,DBLP:journals/iacr/GayP20}.

\newcommand{\Sim}{\mathsf{Sim}\xspace}
\newcommand{\Gen}{\mathsf{Gen}\xspace}
\renewcommand{\Enc}{\mathsf{Enc}\xspace}
\newcommand{\Eval}{\mathsf{Eval}\xspace}
\renewcommand{\Dec}{\mathsf{Dec}\xspace}
\newcommand{\Adv}{\mathsf{Adv}\xspace}

\begin{definition}[Representation model ({\cite[Definition 2]{DBLP:conf/tcc/IshaiP07}})]
    \emph{A representation model} is a polynomial-time computable function
    $U\colon\B^*\times\B^*\to\B^*$, where $U(P,x)$ is referred to as
    the value returned by a ``program'' P on the input $x$.
\end{definition}

\begin{definition}[Computing on encrypted data ({\cite[Definition 5]{DBLP:conf/tcc/IshaiP07}})]\label{def:computing-on-encrypted-data}
    Let $U:\{0,1\}^*\times\{0,1\}^*\to\{0,1\}^*$ be a polynomial-time computable function.
    A \emph{protocol for evaluating programs from $U$ on encrypted data}
    is defined by a tuple of algorithms $(\Gen,\Enc,\Eval,\Dec)$ and
    proceeds as follows.
    \begin{description}
        \item[\textsc{Setup}]
            Given a security parameter $k$, the client computes $(\PK,\SK)\gets\Gen(1^k)$
            and saves $\SK$ for a later use.
        \item[\textsc{Encryption}]
            The client computes $c\gets\Enc(\PK,x)$, where $x$ is the input on which a program $P$
            should be evaluated.
        \item[\textsc{Evaluation}]
            Given the public key $\PK$, the ciphertext $c$, and a program $P$,
            the server computes an encrypted output $c'\gets\Eval(1^k,\PK,c,P)$.
        \item[\textsc{Decryption}]
            Given the encrypted output $c'$, the client outputs $y\gets\Dec(\SK, c')$.
    \end{description}
    We require that if both parties act according to the above protocol,
    then for every input $x$, program $P$, and security parameter $k\in\N$,
    the output $y$ of the final decryption phase is equal to $U(P,x)$ except,
    perhaps, with negligible probability in $k$.
\end{definition}

\begin{definition}[Client privacy ({\cite[Definition 6]{DBLP:conf/tcc/IshaiP07}})]\label{def:client-privacy}
    Let $\Pi = (\Gen, \Enc, \Eval, \Dec)$ be a protocol for computing on encrypted data.
    We say that $\Pi$ satisfies the \emph{client privacy} requirement
    if the advantage of any probabilistic polynomial time (PPT) adversary $\Adv$
    in the following game is negligible in the security parameter $k$:
    \begin{itemize}
        \item $\Adv$ is given $1^k$ and generates a pair $x_0,x_1\in\{0,1\}^*$ such that $|x_0|=|x_1|$.
        \item Let $b\Rgets\{0,1\}$, $(\PK,\SK)\gets\Gen(1^k)$, and $c\gets\Enc(\PK,x_b)$.
        \item $\Adv$ is given the challenge $(\PK, c)$ and outputs $a$ guess $b'$.
    \end{itemize}
    The advantage of $\Adv$ is defined as $\mathbf{Pr}[b=b']-1/2$
\end{definition}

\begin{definition}[Size hiding server privacy: honest-but-curious model ({\cite[Definition 7 and Definition 8]{DBLP:conf/tcc/IshaiP07}})]\label{def:size-hiding-server-privacy-semi-honest}
    Let $\Pi = (\Gen,\Enc,\Eval,\Dec)$ be a protocol for evaluating programs
    from a representation model $U$ on encrypted data. We say that $\Pi$ has
    \emph{computational server privacy in the honest-but-curious model}
    if there exists a PPT algorithm $\Sim$ such that the following holds.
    For every polynomial-size circuit family $D$, there is a negligible function $\epsilon(\cdot)$
    such that for every security parameter $k$,
    input $x\in\{0,1\}^*$, pair $(\PK, c)$ that can be generated by $\Gen,\Eval$ on inputs $k,x$, and
    program $P\in\{0,1\}^*$,  we have
    \[
    \mathrm{Pr}[D(\Eval(1^k,\PK,c,P)) = 1] - \mathrm{Pr}[D(\Sim(1^k,1^{|x|},\PK,U(P,x)))=1]\le \epsilon(k).
    \]
\end{definition}

\begin{definition}[Size hiding server privacy: fully malicious model ({\cite[Definition 12 and Definition 13]{DBLP:conf/tcc/IshaiP07}})]\label{def:size-hiding-server-privacy-malicious}
    Let $\Pi=(\Gen,\Enc,\Eval,\Dec)$ be a protocol for evaluating programs from a representation
    model $U$ on encrypted data. We say that $\Pi$ has \emph{computational server privacy
    in the fully malicious model} if there exists a computationally unbounded, randomized algorithm $\Sim$
    such that the following holds.
    For every polynomial-size circuit family $D$, there is a negligible function $\epsilon$
    such that for every security parameter $k$, arbitrary public key $\PK$,
    and arbitrary ciphertext $c^*$,
    there exists an ``effective'' input $x^*$ such that for every program $P\in\{0,1\}^*$, we have
    \[
    \mathrm{Pr}[D(\Eval(1^k,\PK,c,P)) = 1] - \mathrm{Pr}[D(\Sim(1^k,\PK,c^*,U(P,x^*)))=1]\le \epsilon(k).
    \]
\end{definition}

\begin{theorem}\label{thm:privacy-for-alice}
    The protocol described in \cref{fig:protocol} provides client privacy
    according to \cref{def:client-privacy}.
\end{theorem}
\begin{proof}[sketch]
    Client privacy of the protocol readily follows from the fact that
    TFHE is an FHE scheme based on Learning-With-Errors problem~\cite{DBLP:journals/joc/ChillottiGGI20,DBLP:journals/jacm/Regev09}.
    \qed{}
\end{proof}

\begin{theorem}\label{thm:privacy-for-bob}
    Assuming the SRL security of TFHE,
    the protocol described in \cref{fig:protocol} provides size hiding server privacy
    against an honest-but-curious client as defined in \cref{def:size-hiding-server-privacy-semi-honest}.
\end{theorem}
\begin{proof}[sketch]
    First, we abstract our protocol after the $i$-th round of execution as 
    $c_i=\Eval(1^k,\PK,\{d_{j}\}_{j=1}^{i}, M, i)$ where
    it holds that decryption of $c_i$ is equal to $M(\sigma'_{1}\cdot\sigma'_{2}\cdots\sigma'_{i})$.

    On inputs $(1^k,1^{|x|},\PK,U(P,x))$,
    we define a simulator $\Sim$, which proceeds as follows:
    \begin{itemize}
        \item $c \gets \Enc_{\PK}(U(P,x))$
        \item Return $c$
    \end{itemize}
    
    By the assumption of the SRL security of TFHE, it holds that
    \[
        c\cequiv c_{i}.
    \]
    Therefore, for any PPT adversary $\mathcal{A}$, we have that
    \begin{align*}
        \Pr[\mathcal{A}(\Eval(1^k,\PK,\{d_{j}\}_{j=1}^{i}, M, i))=1]
        -\Pr[\mathcal{A}(\Sim(1^k,1^{|x|},\PK,U(P,x))=1)]\\
        = \epsilon(k)
    \end{align*}
    and the theorem follows.
    \qed{}
\end{proof}

\begin{theorem}\label{thm:privacy-for-bob-malicious}
    Assuming the SRL security of TFHE and the honest generation of the public key $\PK$,
    the protocol described in \cref{fig:protocol} provides size hiding server privacy
    against a malicious client as defined in \cref{def:size-hiding-server-privacy-malicious}.
\end{theorem}
As noted in~\cite{DBLP:conf/eurocrypt/DucasS16}, 
a malicious client may try  to generate invalid ciphertexts or public keys to gain an advantage 
against the DFA held by the server. 
Fortunately, our protocol can easily achieve size hiding server privacy against a malicious client 
(i.e., Alice) by ensuring the honest generation of the public key $\PK$ in combined with 
the SRL security of TFHE~\cite{DBLP:conf/eurocrypt/DucasS16}, and we omit a formal proof for
Theorem~\ref{thm:privacy-for-bob-malicious}.

\section{TFHE Parameters}\label{appendix:TFHE_parameters}
The parameters for TFHE are the foundation of the security of our proposed protocols and greatly affect the performance. 
The security of the parameters is estimated by using lwe-estimator~\cite{AlbrechtPlayerScott+2015+169+203}. 
We use the default parameter provided by TFHEpp~\cite{DBLP:conf/uss/MatsuokaBMS021}. 
This parameter is selected to maximize the performance of Bootstrapping while achieving 128-bit security.
In \cref{tab:params}, we show all necessary parameters used in our implementation and briefly explain the meaning of each parameter.
The parameters which directly affect the security guarantee are $q, \underline{N},\underline{\alpha}, N, \alpha, \overline{N}$, and $\overline{\alpha}$. 
The example of the estimation code for the security of TFHE using lwe-estimator is given at~\cite{estcode}. 
In the following subsections, we briefly explain some ideas which are omitted in the main text for simplicity but are necessary to understand the meaning of parameters.

\begin{table}[bt]
    \centering\scriptsize
    \caption{Table of TFHE parameters}
    \label{tab:params}
    \begin{tabular}{c|c|c}\toprule
         Parameter& Value in implementation & Meaning   \\\midrule
         $q$ & $2^{32}$ & The modulus for discretizing Torus for lvl0 and lvl1\\
         $\overline{q}$ & $2^{64}$ & The modulus for discretizing Torus for lvl2\\
         $\underline{N}$& $635$ & The length of the \TLWElvlz ciphertext\\
         $\underline{\alpha}$ & $2^{-15}$ & \begin{tabular}{c}The standard deviation of the noise for\\ the fresh \TLWElvlz ciphertext\end{tabular}\\
         $N$ & $2^{10}$ & \begin{tabular}{c}
              The length of the \TLWElvlo ciphertext \\
              and the dimension of \TRLWElvlo ciphertext\\
         \end{tabular}\\
         $\alpha$ & $2^{-25}$ & \begin{tabular}{c}The standard deviation of the noise for\\ the fresh \TLWElvlo, \TRLWElvlo and \TRGSWlvlo ciphertext\end{tabular}\\
         $\overline{N}$& $2^{11}$ & \begin{tabular}{c}
              The length of the \TLWElvlt ciphertext \\
              and the dimension of \TRLWElvlt ciphertext\\
         \end{tabular}\\
         $\overline{\alpha}$ & $2^{-44}$ & \begin{tabular}{c}The standard deviation of the noise for \\ the fresh \TLWElvlt, \TRLWElvlt and \TRGSWlvlt ciphertext\end{tabular}\\
         $l$ & 3 & Half of the number of rows in \TRGSWlvlo  \\
         $Bg$ & $2^{6}$ & The base for \CMux with \TRGSWlvlo\\
         $\overline{l}$ & 4 & Half of the number of rows in \TRGSWlvlt  \\
         $\overline{Bg}$ & $2^{9}$ & The base for \CMux with \TRGSWlvlt\\
         $t$ & 7 & The number of digits in \IdentityKeySwitching \\
         $base$ & $2^{2}$ & The base for \IdentityKeySwitching \\
         $\overline{t}$ & 10 & The number of digits in \PrivateKeySwitching \\
         $\overline{base}$ & $2^{3}$ & The base for \PrivateKeySwitching \\\bottomrule
    \end{tabular}
\end{table}

\subsection{Discretization of Torus}
TFHE uses Torus, $\mathbb{T}=\mathbb{R}/\mathbb{Z}$, i.e., the set of real number modulo 1, as one of the most fundamental primitives, but lwe-estimator can treat only $\mathbb{Z}_q$, i.e., the set of integers modulo $q$. 
Therefore, in actual implementation, we discretized Torus into $q$ parts ($\mathbb{T}_q$). 
By this discretization, we can reduce the security of \TLWE, \TRLWE, and \TRGSW into standard LWE, RLWE, and RGSW~\cite{cryptoeprint:2021:1402}. Therefore, we can use lwe-estimator to estimate the security of TFHE. 
In addition to that, because floating-point operations are generally slower than integer operations, this discretization also gives a performance benefit.
This is why we need the parameter $q$.

\subsection{Levels in TFHE}
In the main text, we introduced only one kind of \TLWE ciphertexts, which can be converted from \TRLWE by \SampleExtract. 
In the real implementation of TFHE, there are three kinds of \TLWE ciphertexts (\TLWElvlz, \TLWElvlo, and \TLWElvlt) and two kinds of \TRLWE ciphertexts (\TRLWElvlo and \TRLWElvlt).
\TLWElvlo and \TRLWElvlo are the ones introduced in the main text as \TLWE and \TRLWE, respectively. 
\TLWElvlz is more compact in ciphertext size than \TLWElvlo but needs more noise to establish 128-bit security. 
More noise means less capability for homomorphic computations. 
Therefore, \TLWElvlz only appears in \Bootstrapping to reduce the complexity. 
\TLWElvlt and \TRLWElvlt are larger in ciphertext size than \TLWElvlo and \TRLWElvlo, respectively. Thus, they need less noise to establish 128-bit security and have more capability for homomorphic computations. 
\TLWElvlt and \TRLWElvlt are used in \CircuitBootstrapping as intermediate representations. 
They can be used in our protocols instead of \TLWElvlo and \TRLWElvlo, but it will cause performance degradation due to their ciphertext size.

Because we can convert \TRLWElvlo and \TRLWElvlt to \TLWElvlo and \TLWElvlt, respectively by \SampleExtract and ``due to the absence of known cryptanalytic techniques exploiting algebraic structure'', it is standard to assume the security of \TRLWElvlo and \TRLWElvlt are the same as \TLWElvlo and \TLWElvlt respectively~\cite{10.1007/978-3-319-98113-0_19}.

\subsection{\IdentityKeySwitching}
\IdentityKeySwitching is one of the homomorphic operations in TFHE. 
This operation converts a \TLWElvlo ciphertext into a \TLWElvlz ciphertext which holds the same plaintext message. 
The parameters $t$ and  $base$ are used in this operation and are not related to security but the performance and the noise growth.
Faster parameters generally mean more noise growth in this operation. Thus, the parameters are selected to balance the performance and the noise growth.

\subsection{\PrivateKeySwitching}
\PrivateKeySwitching is one of the homomorphic operations in TFHE.
This operation converts a \TLWElvlt ciphertext into a \TRLWElvlo ciphertext which holds the result of applying the private Lipschitz (linear) function to the plaintext of the input \TLWElvlt ciphertext. 
In \CircuitBootstrapping, we need to apply the function which depends on the part of the secret key. Thus, \PrivateKeySwitching is used to hide the function to avoid leaking a part of the secret key.
The parameters $\overline{t}$ and $\overline{base}$ are used in this operation and not related to security, but the performance and the noise growth. 

\subsection{Parameters for \TRGSW}
In the main text, we only introduced \TRGSWlvlo as \TRGSW, but we also use \TRGSWlvlt in \CircuitBootstrapping.
The parameter $l$ and $\overline{l}$ determine the shape of \TRGSWlvlo and \TRGSWlvlt ciphertexts, respectively, and $Bg$ and $\overline{Bg}$ are used in \CMux with \TRGSWlvlo and \TRGSWlvlt ciphertexts, respectively. Therefore, they highly affect the performance of \CMux. 
A \TRGSWlvlo and \TRGSWlvlt ciphertext can be seen as the $2l$ and $2\overline{l}$ dimensional vector of \TRLWElvlo and \TRLWElvlt ciphertexts, respectively. 
Though each row encrypts the plaintext which is a multiple of the other row's plaintext, there are no known cryptanalytic techniques exploiting this fact. 
Therefore, the security of \TRGSWlvlo and \TRLWElvlt ciphertexts are assumed to be the same as \TRLWElvlo and \TLWElvlo, \TRLWElvlt and \TLWElvlt, respectively. 
Increasing $l$ and $\overline{l}$ means increasing the number of polynomial multiplications in \CMux, which are the heaviest operation in \CMux, but exponentially reduce the noise growth in \CMux. 
$Bg$ and $\overline{Bg}$ are related to noise growth only, and there is the optimal value for fixed $l$ and $\overline{l}$. 
Therefore, $l, Bg, \overline{l}$ and $Bg$ are selected to balance between the performance and the noise growth.

\section{Detailed Experiment Results and Discussion}\label{appendix:other_experiments}

\Cref{tbl:detailed-experiment-results-M_m-fixed-states} shows
the detailed results of $M_m$ when the number of states is fixed.
We observe that the memory usage is constant to the number of the monitored ciphertexts.
We also observe that, in both algorithms, the runtimes of \CMux{} and \CircuitBootstrapping{}
are linear to the length $n$ of the monitored ciphertexts.
These observations coincide with the complexity analysis in \cref{subsec:complexity_analysis}.
In contrast, we observe that the runtimes of \Bootstrapping{} do not change
so much from $n=30000$ to $50000$.
This is because we set $\BootstrappingInterval=30000$, and
\crefrange{line:reversed:for00}{line:reversed:for01} in \cref{alg:reversed} are executed only once.

\Cref{tbl:detailed-experiment-results-M_m-fixed-inputs} shows
the detailed results of $M_m$ when the number of the monitored ciphertexts is fixed.
The table shows that, in both algorithms, the runtimes of \CMux{} and \Bootstrapping{} are
linear to the number $m$ of the states,
and the runtimes of \CircuitBootstrapping{} are sublinear to $m$.
Since, as described in \cref{sec:experiment}, the number of states of the reversed DFA of $M$ is
equal to that of $M$, these observations coincide with the complexity analysis in \cref{subsec:complexity_analysis}.

In both \cref{tbl:detailed-experiment-results-M_m-fixed-states} and \cref{tbl:detailed-experiment-results-M_m-fixed-inputs},
we observe that the memory usage of \BlockStream is larger than that of \ReverseStream.
This is because \CircuitBootstrapping needs a larger bootstrapping key than \Bootstrapping{},
and we need to place the key on the memory when \CircuitBootstrapping{} is performed.

\Cref{tbl:detailed-experiment-results-blood-glucose} shows
the detailed results of blood glucose monitoring.
We observe that, when we use \ReverseStream,
the amounts of memory used for $\psi_1$ and $\psi_2$ are
much larger than those for $\phi_1,\phi_4,\phi_5$.
This is because the numbers of states of the reversed DFA of $\psi_1,\psi_2$
are much larger than those of $\phi_1,\phi_4,\phi_5$.

\begin{table}[tb]
    \centering\scriptsize
    \caption{Experimental results of $M_m$ when the number of the states (i.e., $m$) is fixed to $500$.}
    \label{tbl:detailed-experiment-results-M_m-fixed-states}
\begin{tabular}{c|c||c|c|c||c|c}\toprule
\multirow{3}{*}{Algorithm} &
\multirow{3}{*}{\begin{tabular}{@{}c@{}}\# of\\Monitored\\Ciphertexts\end{tabular}} &
\multicolumn{4}{c|}{Runtime (s)} &
\multirow{3}{*}{\begin{tabular}[c]{@{}c@{}}Memoery\\ Usage\\ (GiB)\end{tabular}} \\ \cline{3-6}
&
&
\multirow{2}{*}{\CMux{}} &
\multirow{2}{*}{\Bootstrapping{}} &
\multirow{2}{*}{\CircuitBootstrapping{}} &
\multirow{2}{*}{Total} & \\
&&&&&&\\
\midrule
\multirow{5}{*}{\ReverseStream}
 & 10000 & 6.94 & --- & --- & 6.98 & 0.34\\
 & 20000 & 13.90 & --- & --- & 13.97 & 0.34\\
 & 30000 & 20.79 & 0.75 & --- & 21.65 & 0.34\\
 & 40000 & 27.64 & 0.83 & --- & 28.63 & 0.34\\
 & 50000 & 34.55 & 0.71 & --- & 35.44 & 0.34\\
\midrule
\multirow{5}{*}{\BlockStream}
 & 10000 & 6.09 & --- & 16.60 & 24.07 & 2.72\\
 & 20000 & 12.33 & --- & 32.20 & 47.19 & 2.72\\
 & 30000 & 18.49 & --- & 47.81 & 70.32 & 2.72\\
 & 40000 & 24.48 & --- & 62.60 & 92.40 & 2.72\\
 & 50000 & 30.88 & --- & 78.71 & 116.11 & 2.72\\
\bottomrule
\end{tabular}
\end{table}

\begin{table}[tb]
    \centering\scriptsize
    \caption{Experimental results of $M_m$ when the number of monitored ciphertexts (i.e., $n$) is fixed to $50000$.}
    \label{tbl:detailed-experiment-results-M_m-fixed-inputs}
\begin{tabular}{c|c||c|c|c||c|c}\toprule
\multirow{3}{*}{Algorithm} &
\multirow{3}{*}{\begin{tabular}{@{}c@{}}\# of\\States\end{tabular}} &
\multicolumn{4}{c|}{Runtime (s)} &
\multirow{3}{*}{\begin{tabular}[c]{@{}c@{}}Memoery\\ Usage\\ (GiB)\end{tabular}} \\ \cline{3-6}
&
&
\multirow{2}{*}{\CMux{}} &
\multirow{2}{*}{\Bootstrapping{}} &
\multirow{2}{*}{\CircuitBootstrapping{}} &
\multirow{2}{*}{Total} & \\
&&&&&&\\
\midrule
\multirow{7}{*}{\ReverseStream}
 & 10 & 10.52 & 0.13 & --- & 10.70 & 0.33\\
 & 50 & 14.60 & 0.43 & --- & 15.14 & 0.33\\
 & 100 & 16.36 & 0.52 & --- & 17.04 & 0.33\\
 & 200 & 21.19 & 0.53 & --- & 21.84 & 0.33\\
 & 300 & 25.95 & 0.64 & --- & 26.72 & 0.33\\
 & 400 & 30.18 & 0.68 & --- & 31.03 & 0.34\\
 & 500 & 34.55 & 0.71 & --- & 35.44 & 0.34\\
\midrule
\multirow{7}{*}{\BlockStream}
 & 10 & 8.02 & --- & 60.70 & 71.35 & 2.71\\
 & 50 & 8.30 & --- & 65.20 & 76.41 & 2.71\\
 & 100 & 10.55 & --- & 73.09 & 87.03 & 2.71\\
 & 200 & 15.86 & --- & 75.38 & 95.50 & 2.71\\
 & 300 & 20.26 & --- & 79.00 & 104.43 & 2.72\\
 & 400 & 25.90 & --- & 79.73 & 111.56 & 2.72\\
 & 500 & 30.88 & --- & 78.71 & 116.11 & 2.72\\
\bottomrule
\end{tabular}
\end{table}

\begin{table}[tb]
    \centering\scriptsize
    \caption{Experimental results of blood glucose monitoring, where $Q$ is the state space of the monitoring DFA and $\Rev{Q}$ is the state space of the reversed DFA.}
    \label{tbl:detailed-experiment-results-blood-glucose}
    
\begin{tabular}{ccccc||c|c|c||c|c|c}\toprule
\multirow{3}{*}{Formula} &
\multirow{3}{*}{\begin{tabular}[c]{@{}c@{}}$|Q|$\end{tabular}} &
\multirow{3}{*}{\begin{tabular}[c]{@{}c@{}}$|\Rev{Q}|$\end{tabular}} &
\multirow{3}{*}{\begin{tabular}[c]{@{}c@{}}\# of blood\\ glucose\\ values\end{tabular}} &
\multirow{3}{*}{Algorithm} &
\multicolumn{4}{c|}{Runtime (s)}&
\multirow{3}{*}{\begin{tabular}[c]{@{}c@{}}Average\\ Runtime\\ (ms/value)\end{tabular}} &
\multirow{3}{*}{\begin{tabular}[c]{@{}c@{}}Memory\\ Usage\\ (GiB)\end{tabular}} \\ \cline{6-9}
&&&&&
\multicolumn{1}{c|}{\multirow{2}{*}{\CMux{}}} &
\multicolumn{1}{c|}{\multirow{2}{*}{\Bootstrapping{}}} &
\multicolumn{1}{c||}{\multirow{2}{*}{\begin{tabular}{@{}c@{}}\textsc{Circuit}\\\textsc{Bootstrapping}\end{tabular}}} &
\multirow{2}{*}{Total} &
&\\
&&&&&
\multicolumn{1}{c|}{} &
\multicolumn{1}{c|}{} &
\multicolumn{1}{c||}{} &
&&\\
\midrule
\multirow{2}{*}{$\psi_1$} & \multirow{2}{*}{10524} & \multirow{2}{*}{2712974}    & \multirow{2}{*}{721}   & \ReverseStream & $16005.55$ & --- & --- & $16021.06$ & $22220.62$ & $42.26$ \\
                          &                        &                             &                        & \BlockStream   & $1.33$ & --- & $105.96$ & $132.68$ & $184.02$ & $2.85$ \\\midrule
\multirow{2}{*}{$\psi_2$} & \multirow{2}{*}{11126} & \multirow{2}{*}{2885376}    & \multirow{2}{*}{721}   & \ReverseStream & $17019.28$ & --- & --- & $17035.05$ & $23626.97$ & $44.92$ \\
                          &                        &                             &                        & \BlockStream   & $1.35$ & --- & $104.96$ & $131.53$ & $182.43$ & $2.86$ \\\midrule
\multirow{2}{*}{$\psi_4$} & \multirow{2}{*}{7026}  & \multirow{2}{*}{---}        & \multirow{2}{*}{721}   & \ReverseStream & --- & --- & --- & --- & --- & --- \\
                          &                        &                             &                        & \BlockStream   & $0.50$ & --- & $20.64$ & $35.42$ & $49.12$ & $2.80$ \\\midrule\midrule

\multirow{2}{*}{$\phi_1$} & \multirow{2}{*}{21}    & \multirow{2}{*}{20}         & \multirow{2}{*}{10081} & \ReverseStream & $21.95$ & $0.24$ & --- & $22.33$ & $2.21$ & $0.33$ \\
                          &                        &                             &                        & \BlockStream   & $17.97$ & --- & $1679.74$ & $1741.15$ & $172.72$ & $2.69$ \\\midrule
\multirow{2}{*}{$\phi_4$} & \multirow{2}{*}{237}   & \multirow{2}{*}{237}        & \multirow{2}{*}{10081} & \ReverseStream & $41.43$ & $0.58$ & --- & $42.23$ & $4.19$ & $0.33$ \\
                          &                        &                             &                        & \BlockStream   & $32.84$ & --- & $1984.61$ & $2073.45$ & $205.68$ & $2.70$ \\\midrule
\multirow{2}{*}{$\phi_5$} & \multirow{2}{*}{390}   & \multirow{2}{*}{390}        & \multirow{2}{*}{10081} & \ReverseStream & $53.50$ & $1.07$ & --- & $54.87$ & $5.44$ & $0.34$ \\
                          &                        &                             &                        & \BlockStream   & $34.12$ & --- & $1988.20$ & $2084.50$ & $206.78$ & $2.70$ \\
\bottomrule
\end{tabular}
\end{table}

\section{Extended Protocol for General Output Interval}

\RB{Mention that $B=\OutputInterval\times|\AP|$.}

In this section, we extend our protocol (\cref{fig:protocol}) to output the monitoring result
after consuming every $\OutputInterval$ Alice's inputs, where $\OutputInterval$ is the interval of the monitoring output.
We present the extended protocol in \cref{fig:protocol-Iout}.
The main difference between \cref{fig:protocol-Iout} and \cref{fig:protocol} is that
the output is generated every after $\OutputInterval$ Alice's inputs consumed
(in \crefrange{line:protocol-Iout:output0}{line:protocol-Iout:output1}).
Since we can prove the correctness and the security of the extended protocol in the same way as the original one,
we omit the proof.

Notice that, in this setting, the block size $B$ for \BlockStream
can be taken to be equal to $\OutputInterval\times |\AP|$.
As the complexity analysis in \cref{subsec:complexity_analysis} implies,
the larger $B$ becomes the faster \BlockStream works.
Therefore, setting $\OutputInterval$ large improves the performance of \BlockStream.

\begin{figure}[tb]
\setlength{\interspacetitleruled}{0pt}%
\setlength{\algotitleheightrule}{0pt}%
\begin{algorithm}[H]\scriptsize
\Input{%
    Alice's private inputs $\sigma_1, \sigma_2, \sigma_3, \dots, \sigma_n\in 2^\AP$,
    Bob's private LTL formula $\phi$,
    $\OutputInterval\in\N^+$,
    and $b\in\B$}
\Output{%
    For every $i\in\mathbb{N}^+$ ($i\le\lfloor n/\OutputInterval\rfloor$),
    Alice's private output which represents 
    $\sigma_1\sigma_2\dots\sigma_{i\times\OutputInterval}\models\phi$}

Alice generates her secret key $\SK$.\;\label{line:protocol-Iout:gen-sk}
Alice generates her public key $\PK$ and bootstrapping key $\BK$ from $\SK$.\;\label{line:protocol-Iout:gen-pk-and-bk}
Alice sends $\PK$ and $\BK$ to Bob.\;\label{line:protocol-Iout:send-pk-and-bk}
Bob converts $\phi$ to a binary DFA $M=(Q, \Sigma=\B, \delta, q_0, F)$.\;\label{line:protocol-Iout:convert-phi}
\For{$i=1,2,\dots,n$}{\label{line:protocol-Iout:for0}%
    Alice encodes $\sigma_i$ to a sequence $\sigma'_i \coloneqq ({\sigma'}_i^1,{\sigma'}_i^{2},\dots,{\sigma'}_i^{|\AP|})\in\B^{|\AP|}$.\;\label{line:protocol-Iout:encode}
    Alice calculates $d_i \coloneqq (\Enc({\sigma'}_i^1),\Enc({\sigma'}_i^2),\dots\Enc({\sigma'}_i^{|\AP|}))$.\;\label{line:protocol-Iout:encrypt}
    Alice sends $d_i$ to Bob.\;\label{line:protocol-Iout:send-ctxt}
    Bob feeds the elements of $d_i$ to \ReverseStream (if $b=0$) or \BlockStream (if $b=1$).\;\label{line:protocol-Iout:feed-ctxt}
    \If{$i\mod \OutputInterval = 0$}{\label{line:protocol-Iout:output0}%
        \tcp{$\sigma'_1\cdot\sigma'_2\cdots\sigma'_i$ refers ${\sigma'}_1^{1}\dots{\sigma'}_1^{|\AP|}{\sigma'}_2^{1}\dots{\sigma'}_2^{|\AP|}{\sigma'}_3^{1}\dots{\sigma'}_i^{|\AP|}$.}
        Bob obtains the output \TLWE ciphertext $c$ produced by the algorithm,
        where $\Dec(c)=M(\sigma'_1\cdot\sigma'_2\cdot\cdots\cdot\sigma'_i)$.\;\label{line:protocol-Iout:obtain-output}
        Bob randomizes $c$ to obtain $c'$ so that $\Dec(c) = \Dec(c')$.\;\label{line:protocol-Iout:randomize}
        Bob sends $c'$ to Alice.\;
        Alice calculates $\Dec(c')$ to obtain the result in plaintext.\;\label{line:protocol-Iout:output1}
    }
}
\end{algorithm}
\caption{Protocol of oblivious online LTL monitoring (extended for the output interval $\OutputInterval$).}\label{fig:protocol-Iout}
\end{figure}

\end{ExtendedVersion}

\end{document}